\documentclass[11pt, lettersize]{article}

\usepackage{amsfonts,setspace}
\usepackage{amsmath, amsthm}
\usepackage{comment}
\usepackage{dsfont}
\usepackage{anysize}
\usepackage{multicol}
\usepackage{graphicx}

\usepackage{tikz}
\usepackage{mathdots}
\usepackage{cancel}
\usepackage{color}
\usepackage{siunitx}
\usepackage{array}
\usepackage{float}
\usepackage{multirow}
\usepackage{gensymb}
\usepackage{tabularx}
\usepackage{extarrows}
\usepackage{booktabs}

\usetikzlibrary{fadings}
\usetikzlibrary{patterns}
\usetikzlibrary{shadows.blur}
\usetikzlibrary{shapes}

\usepackage{enumerate}
\usepackage{enumitem}
\usepackage{color}
\usepackage{algorithm}
\usepackage{algpseudocode}
\usepackage{caption}
\usepackage{subcaption}

\usepackage{hyperref}
\usepackage{cleveref}

\usepackage[skins]{tcolorbox}
\tcbuselibrary{breakable}
\usepackage{titlesec}

\newtcolorbox[
  auto counter, number within=section
]{algorithmbox}[3][]{%
  enhanced,
  colback=white,colframe=black,coltitle=black,
  sharp corners,boxrule=0.4pt,
  fonttitle=\itshape,
  attach boxed title to top left={yshift=-0.3\baselineskip-0.4pt,xshift=2mm},
  boxed title style={tile,size=minimal,left=0.5mm,right=0.5mm,
                     colback=white,before upper=\strut},
  float*=htb,                         
  title   ={Protocol~\thetcbcounter: #2},                       
  label   ={#3},
  #1                                   
}

\crefname{algorithmbox}{Protocol}{Protocols}

\newcommand{\abs}[1]{|#1|}
\theoremstyle{plain}
\newtheorem{theorem}{Theorem}
\newtheorem{lemma}{Lemma}
\newtheorem{proposition}{Proposition}

\newtheorem{fact}[theorem]{Fact}

\theoremstyle{definition}
\newtheorem{definition}{Definition}

\newtheorem{claim}{Claim}
\newtheorem{property}{Property}

\theoremstyle{remark}

\newtheorem{remark}{Remark}

\newcommand{\N}{\mathbb{N}}

\newcommand{\bigo}[1]{O\!\left(#1\right)}

\newcommand{\nonthreecol}{\textsc{Non-3-Colorability}}
\newcommand{\nonkcol}{\textsc{Non-k-Colorability}}

\newcommand{\subgraphcounting}{\textsc{subgraph-counting}($H$)}

\newcommand{\Unif}{\mathrm{Unif}}

\newcommand{\TV}[2]{\textrm{TV}\left(#1,#2\right)}
\newcommand{\TVc}[2]{\mathrm{TV}\left(#1,#2\right)}
\newcommand{\fieldq}{\mathbb{F}_q}
\newcommand{\field}{\mathbb{F}}
\newcommand{\Prob}[1]{\mathbb{P}\left[#1\right]}

\newcommand{\pal}[3]{#1\left[\alpha_{#2}^{#3}\right]}

\newcommand{\pbet}[3]{#1\left[\beta_{#2}^{#3}\right]}

\newcommand{\id}{\textsc{id}}

\newcommand{\sumcheckproblem}{\textsc{Sumcheck-problem}}

\newcommand{\bigoo}{\mathcal O}

\usepackage{authblk} 
\author[1,2]{Benjamin Jauregui}
\author[3]{Masayuki Miyamoto}

\affil[1]{Universidad de Chile, Chile}
\affil[2]{IRIF, Université Paris Cité, France}
\affil[3]{University of Tsukuba, Japan}

\begin{document}

\title{
Distributed Statistical Zero-Knowledge Proofs via Sumcheck
}

\date{\empty}

\maketitle

 \begin{abstract}
We study distributed zero-knowledge proofs, introduced by Bick, Kol, and Oshman (SODA 2022). Although distributed interactive proofs (dIP) have progressed rapidly in recent years, general-purpose techniques for distributed zero-knowledge proofs remain limited and are often problem-specific. We address this gap by introducing distributed statistical zero-knowledge, requiring that each node’s view be simulatable within a negligible statistical distance, and by using this notion to lift the robust Sumcheck protocol (Lund, Fortnow, Karloff and Nisan, FOCS’90) into a modular primitive for constructing distributed zero-knowledge proofs.

Our main technical contribution is a distributed zero-knowledge implementation of the Sumcheck protocol. Given oracle access to a global polynomial $F$ over a finite field $\mathbb F$ with $N$ variables, we design a protocol to verify claims of the form
$\sum_{x\in\field} F(x) = a$ using $\bigo{N}$ rounds of $\bigo{\log|\mathbb F|}$-bit messages, while achieving the statistical zero-knowledge property and small soundness error. We then apply this primitive for two families of problems:
\begin{itemize}
    \item \textbf{\nonkcol}: An $\bigoo(n)$-round, $\bigoo(\log^{1+o(1)} n)$-bit distributed statistical zero-knowledge proof deciding if a graph is not $k$-colorable, for any constant $k$. This is the first nontrivial distributed interactive proof for this problem (even for the case $k=3$, or without zero-knowledge guarantees), with total $\tilde{\bigoo}(n)$ per-node communication with the prover, contrasting with the $\tilde{\Omega}(n^2)$ lower bound for non-interactive proofs (Göös--Suomela, DC~2016).
    \item \textbf{Subgraph Counting}: An $\bigoo(k\log n)$-round, $\bigoo(k\log n)$-bit distributed statistical zero-knowledge proof to count the number of induced/non-induced copies of a given $k$-node pattern $H$ in a graph. For a superconstant $k$, this is the first distributed interactive proof. For constant $k$, our protocol improves the round and message complexity of the distributed interactive proof of Naor--Parter--Yogev (SODA~2020) from $\mathrm{polylog}(n)$ to almost $\bigoo(\log n)$, while additionally providing a statistical zero-knowledge property. We note that the only previously known distributed zero-knowledge results concerning subgraph counting focus on \emph{triangle-freeness} in a 1-round protocol (Grilo--Paz--Perry, 2026), which requires $\tilde{O}(\sqrt{n})$-bit message per node.
\end{itemize}
We further show that beyond the baseline $\bigoo(N)$-round Sumcheck protocol, additional round compression of the Sumcheck protocol is problem-dependent. For Subgraph Counting with constant $k$, a divide-and-conquer approach, recently developed for interactive oracle proofs (Levrat--Medevielle--Nardi, CiC~2025) yields a round complexity of $\bigoo(\log\log n)$, at the cost of slightly larger per-node messages of size $\bigoo(\log^{2+o(1)} n)$ per round.
In contrast, for the problem \nonthreecol\ on constant-degree graphs, we obtain an $\bigoo(n/\log n)$-round protocol and prove a conditional lower bound ruling out $o(n/\log n)$ rounds even on constant-degree graphs when each node is restricted to polynomial-time computation.

\end{abstract}

\tableofcontents
\newpage

\section{Introduction}
Interactive proofs (IP) \cite{Babai1985, goldwasser1989knowledge} have been a key computational model studied in complexity theory since their introduction. In this model, an interactive proof (IP) is a protocol in which two parties, a prover and a verifier, engage in alternate rounds of communication to establish the validity of a statement. The prover is assumed to have unbounded computational power and aims to convince the verifier that the statement is true. The verifier is limited to efficient (typically polynomial-time) computation and uses randomness in its decisions. If the statement is true, the prover can convince the verifier with high probability, and if it is false, no matter what strategy a malicious prover follows, it will fail to convince the verifier except with small probability. One of the most famous IP is the Sumcheck protocol, introduced by Lund, Fortnow, Karloff, and Nisan~\cite{lund1992algebraic}. This is an interactive protocol that allows the verifier to check in polynomial time an equality of the form $\sum_{x \in A} F(x) = a$, where $F$ is a $N$-variate polynomial over a finite field, %that usually depends on some input $n$-node graph $G$\footnote{For example, $F$ can be the arithmetization of 3-coloring of a graph $G$}
and the set $A$ has size exponential in $N$ (typically, a boolean hypercube $\{0,1\}^N$). This question arose in the context of placing $\#\mathsf{P}$ within interactive proofs. The sumcheck protocol revealed that the power of interactive proofs is far greater than originally believed, ultimately leading to Shamir’s landmark result $\mathsf{IP} = \mathsf{PSPACE}$~\cite{shamir1992ip}. Since its inception, Sumcheck\ has found a wide range of applications, from verifiable delegation of computation~\cite{goldwasser2015delegating,reingold2016constant} to real-world blockchain systems~\cite{xie2019libra,ben2019aurora}.

While interactive proofs allow one to convince the verifier (with high probability) that a statement is true, they may also reveal \emph{why} it is true, namely the witness or some information that may be impossible to obtain by the verifier (in polynomial time) without the help of the prover. Goldwasser, Micali, and Rackoff~\cite{goldwasser1989knowledge} introduced \emph{zero-knowledge proofs} to study such leakage. As in any interactive proof, a powerful prover tries to convince a computationally bounded verifier of the validity of a claim through a sequence of message exchanges. What makes zero-knowledge proofs distinctive is that, even after the entire interaction, the verifier cannot gain any information that could not have been generated independently (in polynomial time), without access to the prover’s knowledge. This ensures that the prover’s “witness” or secret remains completely hidden, while still enabling the verifier to be convinced of the statement’s correctness. In this paper, we adapt the Sumcheck\ protocol to the distributed setting and present a distributed zero-knowledge variant of it.

\subsection{Distributed verification and zero-knowledge.}
In recent years, interactive proofs have been extended to the distributed setting.
A distributed interactive proof (dIP), introduced by Kol, Oshman, and Saxena~\cite{kol2018interactive}, enables the nodes of a graph $G$ which can communicate only with their immediate neighbors, to verify a global predicate $\mathcal{P}$ about $G$ with the assistance of an untrusted prover. Communication proceeds in synchronous and alternating rounds between the nodes and the prover. After the interaction ends, each node executes a verification algorithm from a fixed class of distributed algorithms $\mathcal A_V$ that determines whether it accepts or rejects. More precisely, in an $r$-round dIP where the prover and each node exchange $\ell$-bit messages in each round, the class $\mathcal A_V$ of verification algorithms will always be a class of distributed algorithms $\mathcal C[\bigoo(1), r\ell]$ that run in $\bigoo(1)$ synchronous rounds, and the neighbors exchange messages of size at most $\bigoo(r\ell)$ bits on every edge in each communication round (and hence $\bigoo(r\ell)$ total communication bits per edge). In each dIP, the values $r$ and $\ell$ are defined according to the problem. (see~\Cref{subsec:dipandzk} for the formal definition).
A dIP $\Pi$ is said to be valid for $\mathcal{P}$ if it satisfies the standard completeness and soundness conditions. Completeness requires that if $G$ satisfies $\mathcal{P}$, then there exists an assignment of messages given by the prover to the nodes, such that with high probability all nodes accept. Soundness requires that if $G$ does not satisfy $\mathcal{P}$, then regardless of the prover’s strategy, with high probability at least one node rejects.

The model generalizes classical proof-labeling schemes (PLS)~\cite{korman2010proof} and locally checkable proofs (LCP)~\cite{goos2016locally}, which correspond to the special case of a single, non-interactive round. The class $\mathcal A_V$ of verification algorithms is limited to $1$-round broadcast CONGEST algorithms where the broadcast message is always its own proof for PLSs, and is $\mathcal C[\bigoo(1), \infty]$ (i.e., constant round LOCAL algorithms) for LCPs. See also \Cref{subsec:related-work} for related work.

\paragraph*{Zero-knowledge in dIP.}  The core intuition behind (the centralized definition) of zero-knowledge is that, on yes-instances, anything that a probabilistic polynomial-time verifier learns from the interaction could already be produced, without interacting with the prover, by a probabilistic polynomial-time simulator. More formally, the (perfect) zero-knowledge property of an IP $\Pi$ requires that for every verifier in the class of probabilistic polynomial time algorithms~$\mathcal{A}_{poly}$, there exists an algorithm~$A \in \mathcal{A}_{poly}$, called a simulator, that samples the verifier’s full view with the \emph{exact} distribution of a real execution of $\Pi$.

In contrast to the centralized setting, in distributed models, an efficient algorithm is typically considered a {\it local} algorithm (ideally, a constant-round distributed algorithm), i.e., anything the nodes exchange with their constant-hop neighbors is deemed legal. With this in mind,  Bick, Kol, and Oshman~\cite{zkdef} defined the notion of distributed zero-knowledge (dZK) proofs: informally, given a (sub)class of constant-round distributed algorithms $\mathcal A_{const}$, a distributed interactive protocol satisfies the zero-knowledge property with respect to the class $\mathcal A_{const}$ if for every yes-instance there exists a distributed algorithm from $\mathcal A_{const}$, a.k.a. a {\it simulator}, whose local output in each node is identical (in distribution) to the real transcript that each node has after the execution of the dIP (see \Cref{subsec:dipandzk} for formal definition). In~\cite{zkdef} the authors provided a three-round dZK protocol for 3-Colorability using $\bigoo(\mathrm{deg}(v))$-bit message per node and an one-round dZK for Spanning Tree Verification using $\bigoo(\log n)$-bit message per node, both having a simulator from the class $\mathcal C[1,\bigoo(\log n)]$. Recently, Grilo, Paz, and Perry~\cite{grilo2025distributed} showed one-round dZK protocols for 3-Colorability using $\bigoo(\log n)$-bit message per node and for triangle-freeness using $\tilde \bigoo(\sqrt n)$-bit message per node, with a simulator from the class $\mathcal C[1,\bigoo(\log n)]$ and $\mathcal C[1,\tilde \bigoo(\sqrt n)]$, respectively.
As more general results, (1) \cite{zkdef} provided a compiler that transforms a proof-labeling scheme (PLS) to a dIP protocol that satisfies the distributed zero-knowledge property defined in \cite{zkdef}, where the simulator class depends on the parameters of the PLS. In particular, the dIP that satisfies the zero-knowledge property significantly increases the proof size, and (2) \cite{grilo2025distributed} showed how to certify any graph property in \textsf{NP}, which is zero-knowledge under some cryptographic assumptions (such as the existence of one-way functions).

Despite the progress, two conceptual gaps remain:
\begin{description}
    \item[Generality.] Existing dZK protocols have some drawbacks: problem-specific, inefficient, or dependent on cryptographic assumptions. There is no general primitive akin to sumcheck that automatically yields efficient protocols that are unconditionally zero-knowledge, applicable for a wider class of properties.
    \item[Perfect vs statistical ZK.] The zero-knowledge property defined in \cite{zkdef} emulates the {\it perfect} zero-knowledge property in IP protocols, where the simulator algorithm needs to simulate {\it exactly the view} (as distribution) that the ZK protocol. While this definition is the main goal that we should ask for a zero-knowledge proof algorithm, it is known in IP protocols that is now known to be achievable for all problems. In classical cryptography and complexity theory, relaxing to statistical zero-knowledge, where we ask the simulator to generate a {\it close enough} view of the original view of the protocol (see \Cref{subsec:dipandzk} for a formal definition), unlocks simpler and more modular constructions. Such a relaxation has not been explored in the distributed setting.
\end{description}

\subsection{Our contributions}
We move beyond previous limitations by lifting the classical sumcheck protocol~\cite{lund1992algebraic} to the distributed domain. In this framework, we study tuples of the form $(G,F,q,a)$ where $G$ is an $n$-node communication graph of the distributed verifier with unique identifiers from the range $\{0,\ldots,n-1\}$, $F:\fieldq^N\rightarrow \fieldq$ is an $N$-variable polynomial of total degree $d$, $q$ is a prime number representing the size of a finite field $\fieldq$, and $a \in \fieldq$ is the claimed target value, where the goal is to decide if $\sum_{\textbf{x}\in\fieldq }F(\textbf{x}) = a$. We call such a tuple a Sumcheck-instance; and in any given Sumcheck-instance $(G,F,q,a)$, each node receives $q$, $a$ as input and has oracle access to $F$.

While it is not hard to obtain a distributed implementation of Sumcheck\ protocol through classical techniques such as aggregations over a spanning tree, we move directly to study how to obtain a distributed implementation of Sumcheck\ protocol that satisfies a {\it zero-knowledge property}. Specifically, we consider the following input: each node has oracle access to $F$, and the input to each node includes $q$, $a$, and by simplicity we also assume the each node receives the ID of the parent in a rooted spanning tree with the ID of the root node is $0$. See \Cref{remark:oracle}.

\begin{theorem}[Distributed ZK Sumcheck]\label{theo:zksc} There exists a statistical zero-knowledge distributed interactive proof such that, for any Sumcheck-instance $(G,F,q,a)$, where $F$ is an $N$-variable polynomial of individual degree $d$ over a finite field $\fieldq$ and any fixed parameter $t \in \mathbb{N}$,  
such that the following holds:

\begin{itemize}
    \item \textbf{Completeness.} If \(\sum_{x \in \{0,1\}^N} F(x) = a\), then all nodes accept with probability \(1\).
    \item \textbf{Soundness.} If \(\sum_{x \in \{0,1\}^N} F(x) \neq a\), then all nodes accept with probability at most \(\tfrac{Nd}{q}+\dfrac{1}{t}\).
\end{itemize}
The protocol runs for \(\bigo{N}\) rounds, with each round requiring \(\bigo{\log q}\) bits of message between the prover and each node. Moreover, the verification algorithm and the simulator belongs to the class $\mathcal C[\bigoo(1),\bigoo(N\log q)]$ requiring a single oracle query to $F$.
\end{theorem}

\begin{remark}[Oracle access $\Leftrightarrow$ aggregation tree]\label{remark:oracle}
    Throughout this work, and as usual in the dIP literature, the target polynomial $F$ is an encoding of some global information of the graph. Usually, each node $v$ can locally evaluate the portion $F_v(x_1,\dots,x_N)$ of $F$ that depends on its private information. In order to evaluate the global value $F(x_1,\dots,x_N)$, the network needs to aggregate\footnote{This aggregation typically is an addition or multiplication of the local portion of $F(x_1,\dots,x_N)$ computed by each node} these partial evaluations along some overlay, which is in practice a rooted spanning tree. Thus, oracle access to $F$ implicitly assumes a rooted spanning tree. Our protocol merely makes this implicit tree explicit. Even if we consider the case where the ID of parents are sent from the prover, verifying its correctness can be done with perfect zero-knowledge~\cite{zkdef}. Throughout the paper we assume the given tree $T$ is rooted at the node with its ID $0$, which is written in the input label of each node. 
\end{remark}

From \Cref{theo:zksc}, we obtain interesting applications that surpass the state-of-the-art results.

\paragraph*{\nonkcol}
The first concrete problem consists of verifying if the graph is not k-colorable, for a fixed integer $k$, which we call \nonkcol\ problem. \Cref{theo:zksc} yields the following result. To keep the notation simple for the moment, we present an informal version of the result. A precise definition of the statistical zero-knowledge dIP class is deferred to \Cref{subsec:dipandzk}, while a fully formal statement can be found in \Cref{theo:finalnon3zk}.
\begin{theorem}[Informal version of~\Cref{theo:finalnon3zk}]\label{theo:informalzkapp1}
    Let $q \in n^{\omega(1)}$ be an arbitrary prime. There exists a statistical zero-knowledge dIP $\Pi$ for \nonkcol\ that requires $\bigo{n}$ rounds and $\bigo{\log q}$ message size, and admits a simulator $\mathcal S \in \mathcal C[\bigoo(1), \bigoo(n \log q)]$ whose output induces, for each node, a distribution statistically close to its view in $\Pi$.
\end{theorem}

Deciding even if a graph is not $3$-colorable is known to be hard in locally checkable proof (i.e., one-round distributed verification): G{\"o}{\"o}s and Suomela~\cite{goos2016locally} showed that any LCP for $\nonthreecol$ requires $\Omega(n^2/\log n)$-bit proof per node. Notably, while most of other known graph properties that require $\tilde{\Omega}(n^2)$ bits in LCPs already admit efficient interactive proofs (symmetric graphs~\cite{kol2018interactive}, asymmetric graphs~\cite{naor2020power}, and graph non-isomorphism~\cite{naor2020power}), there is no interactive proofs for $\nonthreecol$ that improves the $\bigoo(n^2)$-bit trivial LCP (giving the adjacency matrix of the graph to each node). Our result improves the total message size per node to almost $\bigoo(n\log n)$. This is also the first zero-knowledge dIP for \nonkcol.

\paragraph*{Subgraph-counting.}
The second problem that we study is $\subgraphcounting$. In this problem, we want to verify that the number of a given $k$-node pattern graph $H$ in the network is exactly $\Delta$, which is given as the input. Using \Cref{theo:zksc}, we obtain the following. 

\begin{theorem}[Informal version of~\Cref{theo:finaltczk}]\label{theo:informalzkapp2}
Let $(H,\Delta)$ be information known to all nodes, where $H$ is a $k$-node graph and $\Delta$ is a natural number (the claimed number of copies of $G$ in the graph). Let $q$ be a prime satisfying $q \in n^{\omega(1)}$ if $k \in \bigoo(1)$, and $q > n^k$ if $k \in \omega(1)$. There exists a statistical zero-knowledge dIP $\Pi$ for \subgraphcounting\ that requires $\bigo{k \log n}$ rounds and $\bigo{\log q}$ message size, and admits a simulator $\mathcal S \in \mathcal C[\bigoo(1), \bigoo(k \log q \log n)]$ whose output induces, for each node, a distribution statistically close to its view in $\Pi$.
\end{theorem}

 For $\subgraphcounting$, our dZK proof improves the total message size per node from~\cite{grilo2025distributed}, where they showed that triangle-freeness admits an one-round dZK proof with $\bigoo(\sqrt{n})$-bit proof against a simulator from $\mathcal C[1,\bigoo(\sqrt{n})]$. Note that triangle-freeness can be solved in 1-round if we allow $\bigoo(n)$-bit message, i.e., algorithms from $\mathcal C[1,\bigoo(n)]$.

\Cref{theo:informalzkapp2} improves the result of~\cite{naor2020power}. In their work, the distributed implementation of the celebrated GKR protocol~\cite{goldwasser2015delegating} is presented. Here, the GKR protocol is a general-purpose interactive proof that enables a verifier to check the correctness of the output of a log-space uniform circuit in quasilinear time. More precisely, they showed that any graph property verifiable by a log-space uniform circuit of size $T=\mathrm{poly}(n)$ and depth $d = \bigoo(\log n)$ admits a distributed interactive proof using $\bigoo(d\cdot \log T \cdot \mathrm{polylog}(n))$ rounds and message size. Since the naive algorithm for subgraph-counting (checking all $k$-tuples of a graph in $\bigoo(n^k)$ time) is implemented by a log-space uniform boolean circuit of size $T=\mathrm{poly}(n)$ and depth $d = \bigoo(\log n)$ when $k\in\bigoo(1)$, it admits a distributed interactive proof of $\mathrm{polylog}(n)$ round and message complexity. Intuitively, the GKR protocol utilizes the sumcheck protocol $\bigoo(d)$ times, resulting in an $\bigoo(d)$ multiplicative factor in the round/message complexity. Therefore, their compiler can be regarded as a black-box implementation of sumcheck. We ``open this black-box", and apply our sumcheck compiler directly, rather than through a black-box GKR layer, to the polynomial encoding the number of subgraphs, in order to bypass the depth-$d$ overhead of the GKR protocol. More importantly, our result is applicable for the $k\in\omega(1)$ regime, where subgraph-counting does not seem accessible through the distributed GKR protocol of~\cite{naor2020power}.~\footnote{The problem of detecting a $k$-clique, a special case of subgraph-counting, is a canonical $\mathsf{W}[1]$-hard problem (i.e., one for which no poly-time algorithm is believed to exist). Moreover, even $n^{o(k)}$-time algorithm violates the exponential-time hypothesis~\cite{cygan2015parameterized}.}

\paragraph*{Optimizing round complexity of \subgraphcounting.} From the first distributed implementation of Sumcheck\ given by \Cref{theo:zksc}, we manage to improve the round complexity using a divide-and-conquer technique inspired from~\cite{levrat2025divide}:

\begin{theorem}[$\bigoo(\log\log n)$-round protocol for $\subgraphcounting$]\label{theo:naivetriangle}
 Let $(H,\Delta)$ be information known to all nodes, where $H$ is a $k$-node graph and $\Delta$ is a natural number (the claimed number of copies of $G$ in the graph). Let $q$ be a prime satisfying $q \in n^{\omega(1)}$ if $k \in \bigoo(1)$, and $q > n^k$ if $k \in \omega(1)$. There exists a statistical zero-knowledge dIP $\Pi$ for \subgraphcounting\ that requires $\bigo{k\log k + \log \log n}$ rounds and $\bigo{\frac{\log^{2} q}{\log k} + k\log k\log q}$ message size, such that the view of each node after the execution of $\mathcal S$ has a distribution statistically close to the distribution of the view of each node in $\Pi$.
\end{theorem}
Notably, for $k\in\bigoo(1)$, this gives exponentially small round complexity of $\bigoo(\log\log n)$ with slightly increased message size of $\bigoo(\log^{2+o(1)}n)$.

\paragraph*{Conditionally optimal round complexity for \nonkcol}
So far we have seen that the polynomial specially designed for $\subgraphcounting$ admits an exponential round compression, i.e., from the naive $\bigoo(\log n)$ to $\bigoo(\log\log n)$. This raises the following natural question: 
\begin{center}
\emph{Can the $\bigoo(n)$ rounds in our \nonkcol\ protocol be improved—perhaps to polylogarithmic complexity?}  
\end{center}
We show that for constant-degree graphs, the round complexity can be reduced to $\bigoo(n/\log n)$, which is essentially optimal under some realistic setting, showing that substantial further reductions are unlikely.

\begin{theorem}\label{theorem:conditional-lower-bound}
\label{thm:n3c-lb}
There is a statistical zero-knowledge dIP protocol for \nonkcol\ on constant-degree graphs that uses $\bigoo(n/\log n)$ rounds and $\bigoo(\log^2 n)$ bits per node per round. In contrast, assume $\mathsf{UNSAT}\notin\mathsf{AMTIME}[2^{o(n)}]$ and let $\Pi$ be any dIP protocol for
$\nonthreecol$ in which each node is a polynomial-time Turing machine.
Then $\Pi$ must use
$
    \Omega\bigl(\frac{n}{\log n}\bigr)
$
rounds even for constant-degree graphs.
\end{theorem}

Here, we condition on the assumption that any two-round Arthur-Merlin protocol with a subexponential-time verifier cannot decide the unsatisfiability of a given formula. This is sometimes called ``Arthur-Merlin ETH", which has been mentioned by several papers~\cite{goldreich2002interactive,carmosino2016nondeterministic,williams2016strong,10.1145/3721134}\footnote{Note that even a faster than $2^n$ two-round
AM protocol would still be a big achievement~\cite{10.1145/3721134}.}. Consequently, our $\bigoo(n/\log n)$-round protocol is optimal unless
(1) major breakthroughs occur in Arthur–Merlin complexity, or (2) the verification algorithm uses exponential-time like in the trivial $\Omega(n^2)$-bit PLS using the adjacency matrix\footnote{In the trivial $\bigoo(n^2)$-bit PLS, each node receives the entire topology of the graph as the adjacency matrix, and then decides its local output by solving $\nonthreecol$ (a $\mathsf{coNP}$-hard problem) in exponential time.}.
On the other hand, \Cref{theorem:conditional-lower-bound} also exhibits that there is not much room in reducing total per-node message complexity of this problem: the trivial PLS (sending all edges in the graph to each node) requires $\Omega(n)$-bit per node for constant-degree graphs, while $\Omega(n^2/\log n)$-bit are required even on the interactive setting.

\subsection{Overview of our techniques}\label{subsec:overview}
We now outline the technical ingredients that drive our results. First of all, in all of our protocols we assume that each node knows the ID of its parent in an arbitrary spanning tree (and for simplicity assume that the root node has ID 0). This is safe, since the spanning tree can be verified using additional $\bigoo(\log n$)-bit in the proofs of the first round by using the classical result of~\cite{korman2010proof} (for non-zero-knowledge protocols) or the dZK proof of~\cite{zkdef}.
\paragraph*{dIP for Sumcheck Protocol.}
The starting point is the Lund–Fortnow–Karloff–Nisan protocol~\cite{lund1992algebraic} (Protocol~\ref{alg:LFKN}). 
\begin{algorithmbox}{\textsc{Sumcheck-Protocol}}{alg:LFKN}
{
\textbf{Input}: A polynomial $F:\mathbb{F}^N\rightarrow \mathbb{F}$, and a field element $a\in \mathbb F$.\\
\textbf{Goal}: Check that $\sum_{\mathbf{x}\in\{0,1\}^N}F(\mathbf{x})=a$.
   \begin{enumerate}
\item \textbf{Step 1.}
\begin{itemize}
\item The prover sends an univariate polynomial \( g_1(x_1) \). 
\item The verifier checks $g_1(0) + g_1(1) = a$, rejects otherwise.
\end{itemize}

\item \textbf{Step 2.} Repeat the following for \( 2 \leq i \leq N \):
\begin{itemize}
\item The verifier picks \( r_{i-1} \in \mathbb{F} \) u.a.r., sends to the prover.
\item The prover sends an univariate polynomial \( g_i(x_i) \).
\item The verifier checks \( g_i(0) + g_i(1) = g_{i-1}(r_{i-1}) \), rejects otherwise.
\end{itemize}

\item \textbf{Final Check}:
The verifier picks \( r_{N} \in_R \mathbb{F} \), checks \( g_N(r_N) = F(r_1,\ldots,r_N) \), rejects otherwise.
\end{enumerate}
}
\end{algorithmbox}

In this protocol, during each prover’s turn \(i \in \{1,\ldots, N\}\), the prover provides an univariate polynomial \(g_i\), where \(N\) is the number of variables in the original polynomial. In our distributed implementation, each polynomial \(g_i\) is represented by \(n\) coefficients, i.e.,
\[g_i(x) = \sum_{j=0}^{n-1} \alpha_j x^j,\]
with each coefficient \(\alpha_j\) stored at a distinct node. A distributed implementation (without zero knowledge guarantees) of the Sumcheck protocol follows essentially from the scheme described in Protocol \ref{alg:LFKN}. Concretely, each evaluation of \(g_i\) is carried out by a converge-cast along a spanning tree, which aggregates the coefficient-weighted monomials given by the prover, and finally the root of the spanning tree is in charge to verify the respective equality in each step. This requires only one round of assistance from the prover: the coefficients are assigned to nodes through the spanning tree, enabling the root to verify the resulting equality. This construction yields the naïve distributed implementation of the Sumcheck protocol (see \Cref{sec:distsc} for details).

\paragraph*{Zero-knowledge framework.} 
The na\"{\i}ve implementation described above reveals coefficients of polynomials, hence each node learns rich structural data. Zero-knowledge therefore, requires hiding \emph{all} intermediate polynomials. For example, consider the task of computing the simple arithmetic sum of $n$ private inputs: each node $v\in \{0,\ldots , n-1\}$ holds a private bit $b_v\in \{0,1\}$, and let $a=\sum_{v=1}^n b_v$ be the value we want to verify. To phrase this in the sumcheck form we take the $N=\log n$-variate multilinear extension $F(x_1,\ldots,x_N)=\sum_{v=0}^{n-1} b_v\cdot \chi_v(x_1,\ldots,x_N)$ where $\chi_v$ is the standard indicator polynomial for the binary representation of $v$. The sumcheck protocol applies to check $\sum_{\boldsymbol{x}\in \{0,1\}^{\log n}}F(\boldsymbol{x}) =a$. Under the distributed implementation described above, the root of the spanning tree performing the convergecast verifies all equalities with high probability. In particular, the root will learn the values $g_1(0)$ and $g_1(1)$ of Step 1, but by definition of Sumcheck protocol , it holds that

\begin{align*}
    g_1(0) &= \sum_{(x_2,\dots,x_N)\in\{0,1\}^{N-1}}F(0,x_2,\dots, x_N) = \sum_{(x_2,\dots,x_N)\in\{0,1\}^{N-1}}b_v\cdot \chi_v(0,\ldots,x_N)\\
    g_1(1) &= \sum_{(x_2,\dots,x_N)\in\{0,1\}^{N-1}}F(1,x_2,\dots, x_N) = \sum_{(x_2,\dots,x_N)\in\{0,1\}^{N-1}}b_v\cdot \chi_v(1,\ldots,x_N)\\
\end{align*}

 In particular, in our function $F(\textbf{x})$, the variable $x_1$ denotes the first bit of the binary representation, and therefore

 \begin{align*}
     g_1(1) &= \sum_{(x_2,\dots,x_N)\in\{0,1\}^{N-1}}b_v\cdot \chi_v(1,\ldots,x_N)\\
     & =\sum_{v:\text{the first bit of }v \text{ is } 1} b_v\\
     &= \text{``the number of nodes with IDs $\ge n/2$ whose bit value $b_v$ is 1’’}.
 \end{align*}

Thus, the single value $g_1(1)$ already reveals the aggregate state of half of the network, which may be located at distance $\Theta(n)$ from the root, and hence leaks information beyond the reach of constant-round simulators.

Therefore, we need to modify the above approach into one that hides all the coefficients. To this end, we use several ideas from distributed and cryptographic algorithms. In the distributed implementation described in \Cref{sec:distsc}, the verification of the equality $\sum_{x\in\{0,1\}^N}F(x) = a$
is reduced to verify $\bigo{N}$ equalities: The equality $g_1(0)+g_1(1)=a$ in Step 1 of Protocol~\ref{alg:LFKN}, and $N-1$ equalities $g_i(0)+g_i(1) = g_{i-1}(r_{i-1})$ in each repetition of Step 2.

While the verification of equalities has been previously studied in distributed implementations of zero-knowledge proofs (such as the verification of a spanning tree in \cite{zkdef}) using the standard random mask approach, where the authors show how the values involved in a single equality can remain hidden from the nodes while still allowing them to verify the equality, we address a more complex scenario. In our case, each coefficient participates in multiple equalities, making the verification process more intricate. 

To illustrate our approach, suppose the nodes of a graph need to verify 

$$g_i(0)+g_i(1) = g_{i-1}(r_{i-1}).$$

Let us now focus on evaluating $g_i(1)$ where $g_i$ is an univariate polynomial of degree $n-1$. An easy protocol could be to ask the prover to provide a single coefficient of $g_i$ to each node $v$ in $G$, and let the network sum them along the spanning tree to compute $g_i(1)$ in the equality. But we want to keep the coefficients hidden from the nodes. In this new scenario, we can use randomness in order to verify the equalities without revealing the values of the addition.

Instead of directly evaluating $g_i$, you can ask the prover to send evaluations disturbed by a random uniform value, namely, to send $\alpha_v +r_v$ for each $v$, where $\alpha_v$ is a coefficient of some monomial in $g_i$, and $r_v$ is an uniform random value over a field $\fieldq$. In this new set up, the nodes now can compute $g_i(1)$ if and only if

\begin{equation}\label{eq:t}
    \sum_{v}r_v = 0.
\end{equation}

The problem is that the nodes should not be able to verify whether Equation~\ref{eq:t} holds, as the direct verification by seeing the masks themselves could leak information about the coefficients: knowing the masked value $\alpha_v + r_v$, and the random mask $r_v$ would reveal the coefficient $\alpha_v$. Therefore, instead of just asking the prover for a masked evaluation, we ask the prover to commit $t$ different possible random masks $r_v^1, \dots, r_v^t$ using Shamir's secret sharing technique, where $t$ is a fixed security parameter. Once the prover has committed to these $t$ values, the nodes randomly choose one index to keep hidden and open the remaining $t-1$ masks. They then check that Equation~\eqref{eq:t} holds for each revealed mask, catching any cheating prover with probability at least $1-1/t$.

This strategy presents several challenges to address. For example, the prover must commit to the random values in such a way that no individual node can recover them without exchanging information with its neighbors—otherwise, a node could compute each random value committed by the prover on its own. Moreover, the nodes must ensure that if they want to verify another equality involving the same polynomial $g_i$, then the prover is required to provide a consistent commitment for its coefficient. We address all these issues in \Cref{subsec:general-toolbox}, and the solutions of these problems represent some of the most technical part of our zero-knowledge implementation. 

To summarize, the high-level idea behind our zero-knowledge implementation is to send each node the original value masked with a random value, in such a way that all the randomness involved in the equality cancels out (i.e., their sum equals zero). This is formalized through the conditions ${\bf (RC)}_i$, $\widetilde{\bf (RC)}_i$, and ${\bf (0C)}_i$ described in Section \ref{subsubsec:random-addition}. The goal is to ensure that these conditions hold without revealing the exact value of any random mask used, while guaranteeing that the same random value is applied consistently to each original value across all equalities in which it appears.

In the applications to \nonkcol\ and $\subgraphcounting$, we need to replace oracle queries, which is not allowed in these problems, by something else. Intuitively, the random evaluations of the polynomials in the verification algorithm can be done using one-round help of the prover. In our simulators, this evaluation is replaced by an uniformly random value, which is sufficient for our purposes, as we prove that, given an uniformly random input, the distributions of the target polynomials are statistically close to the uniform distribution. For the polynomial $P_G$ in $\nonthreecol$, the evaluation is essentially close to the product of uniform random variables, which itself follows a distribution that is statistically close to the uniform distribution. For the polynomial $f$ in $\subgraphcounting$, we can use the fact that $f$ is a multilinear polynomial, for which we show that the distribution is almost uniform. Moreover, since the masks used for our statistical zero-knowledge protocols are component-wise, each oracle call to the intermediate polynomials used in our $\bigoo(\log\log n)$-round dIP protocol for $\subgraphcounting$ does not leak any information other than its evaluation value. Using the fact that every intermediate polynomial is also multilinear, we also get a statistical zero-knowledge variant of our $\bigoo(\log\log n)$-round dIP.

\paragraph*{\nonkcol.} We exemplify our techniques for the \nonthreecol\ problem, since its generation to \nonkcol\ follows naturally. 

 To instantiate our framework, given an instance $(G=(V,E),I)$, where $G$ is the communication graph which we want to verify if it is non-$3$-colorable and $I$ denotes the input label to each node such as the IDs of itself and its parent in a spanning tree, we use a polynomial $P_G\colon \{0,1\}^{3n}\to \N$ of $3n$ variables, such that for each node $v\in V$ the polynomial $P_G$ has three variables $r_v,b_v$ and $g_v$, and the polynomial $P_G(x) = T(x)\cdot S(x)$ is defined as follows.

First,$T$ is defined as follows. For each edge \(e=(u,v)\in E\) with \(u<v\), define
\begin{align*}
   W^{\text{red}}_e   := 1 - r_u r_v, 
\qquad W^{\text{green}}_e := 1 - g_u g_v, 
 \qquad  W^{\text{blue}}_e  := 1 - b_u b_v.
\end{align*}
Then, be defining $
   T_{\text{red}} := \prod_{e\in E} W^{\text{red}}_e$, and $ 
   T_{\text{green}}$, $ T_{\text{blue}}$ analogously, we define $
   T := T_{\text{red}} T_{\text{green}} T_{\text{blue}}.
$

On the other hand, for each node \(v \in V\), define
\begin{align*}
   A_v &:= 1 - (1-r_v)(1-g_v)(1-b_v), \\
   B_v &:= (1-r_v g_v)(1-g_v b_v)(1-r_v b_v).
\end{align*}
And then
$
   S := \prod_{v\in V} A_v B_v.
$

Our polynomial construction builds on the classical reduction of 3-Coloring to 3-SAT: the resulting polynomial is the arithmetization of the reduced 3-CNF formula. Intuitively, the factor \(T\) enforces edge-consistency of the coloring, while \(S\) guarantees that each node is assigned a unique color. The sum of the constructed target polynomial \(P\) then encodes the number of valid 3-colorings. Consequently, the underlying graph is non-3-colorable if and only if the following equality holds:
\begin{equation}
\label{eq:1}
    \sum_{x \in \{0,1\}^{3n}} P_G(x) = 0 .
\end{equation}
That is, verifying that the sum of \(P_G\) over all assignments \(x \in \{0,1\}^{3n}\) equals zero proves that the graph is non-3-colorable.

Therefore, the Sumcheck protocol from \Cref{theo:zksc} reduces the verification of \nonthreecol of the input graph to deciding whether the Sumcheck-instance \((G, P_G, poly(n), 0)\) satisfies Equation~\eqref{eq:1}. The evaluation of \(P_G\) at a single point is performed through a spanning tree, using partial evaluations provided by the prover in a single round. 

Moving from \nonthreecol\ to \nonkcol\ for any fixed integer $k$ follows naturally: We reproduce the same reduction to obtain the polynomial $P_G(x)$ as described above, but now this polynomial has $k\cdot n$ variables, and the rest of the protocol proceeds in the same way.

A central technical challenge in our application of the Sumcheck protocol
is that the size $N$ of the sum may be exponential in the
number of nodes of the graph. Indeed, even in simple instances, the number of
valid $K$-colorings can grow exponentially with the size of the graph; for
example, a path on $n$ nodes admits at least $3\cdot 2^{n-1}$ distinct
$3$-colorings. Consequently, the value handled by the protocol cannot, in
general, be represented explicitly using a polynomial number of bits.

However, in our setting, the protocol requires that every intermediate value
communicated between nodes be encoded using only $\bigo{\log n}$ bits,
independently of the prover's strategy. To achieve this, we rely on the Prime
Number Theorem and perform the computation modulo a suitably chosen prime $q$.
This allows all intermediate quantities to remain logarithmic in size.
Moreover, by selecting $q$ from an appropriate range, we ensure that, with
high probability, the protocol still correctly certifies the property, since
distinct values collide modulo $q$ only with negligible probability.

\paragraph*{Subgraph-Counting.} For \subgraphcounting, we describe the high-level idea for the case of triangles.
Triangle counting is a well-studied problem in the centralized setting. Given a graph $G$, the number of triangles of $G$, denoted as $\#\Delta(G)$,  can be obtained by matrix multiplication: 
\[
\#\Delta(G) = \frac16\mathrm{tr}(A^3),
\]
where $\mathrm{tr}$ is the trace of a matrix and $A$ is the adjacency matrix of the graph $G$.
To rephrase this relation by a sum of a polynomial, the adjacency matrix $A$ is viewed as a boolean function, then it is extended to a unique multilinear polynomial $\widetilde{A}$ over a large field. Now we get a polynomial $\widetilde{A}$ of $2\lceil \log n\rceil$ variables. The final polynomial $f$ is defined as 
\[
f(i,j,k)=\widetilde{A}(i,j) \cdot \widetilde{A}(j,k) \cdot \widetilde{A}(k,i),
\]
which is a $3\lceil \log n\rceil$-variate polynomial of individual degree 2. By known facts about multilinear extensions, it is known that the sum $\displaystyle\sum_{i,j,k\in\{0,1\}^{\log n}}f(i,j,k)$ is equal to $\mathrm{tr}(A^3)$, and therefore we can use the Sumcheck protocol to compute the number of triangles.
Now our Sumcheck compiler reduces the computation of $\sum_{i,j,k}f(i,j,k)$ to a single evaluation of $f$, using $\bigoo(\log n)$-rounds of interaction. Later in our zero-knowledge simulator, we check that the evaluation can be replaced by an uniform random value, as an evaluation on a random point is nearly uniform.

\paragraph*{Reducing the round complexity of Sumcheck for triangle counting.}

In the above triangle counting protocol, the number of triangles is encoded in the sum of a $\bigoo(\log n)$-variate polynomial. During each round of the Sumcheck, the intermediate univariate polynomial has $\bigoo(\log n)$ monomials, and oracle access to it is realized by distributing the coefficients of monomials over the network. During each oracle call, the other nodes sit idle, wasting parallel bandwidth of the network. We remedy this by \emph{inflating} the intermediate polynomial:
borrowing the idea of~\cite{levrat2025divide}, we eliminate a constant ratio of variables in each round, instead of a single variable as in the standard Sumcheck. 
The idea of~\cite{levrat2025divide} is summarized as follows. In the original Sumcheck protocol, at each round $2 \le i \le N$, the protocol considers the Sumcheck relation
$a_i = \sum g_i$ over $N-i+1$ variables, where
\[
a_i := \underset{x_i,\ldots,x_N \in \{0,1\}}{\sum} F(r_1,\ldots, r_{i-1}, x_i,\ldots,x_N)
\]
is implicitly given as a polynomial $g_{i-1}$ whose correctness verified by the previous rounds (together with the final check).
Therefore, each round is typically viewed as eliminating a single variable by replacing it with a random value. Instead, we eliminate $N/2$ variables in a single round, reducing the original Sumcheck relation $a=\sum F$ to a new relation $a'=\sum F'$, where
\[
F'(\cdot)=\underset{x_{N/2+1,\ldots,x_N \in \{0,1\}}}{\sum} F(\cdot,x_{N/2+1},\ldots, x_N)
\]
has $N/2$ variables, and $a'$ is defined by 
\[
a'=a.
\]
In this case, the sum remains the same and we don't need to verify its correctness. Instead, we need to verify the correctness of a given polynomial $F'$. To this end, we define another polynomial $F''$ of $N/2$ variables
\[
F''(\cdot)=F(r_1,\ldots,r_{N/2},\cdot)
\]
for random values $r_1,\ldots,r_{N/2}$ and introduce another Sumcheck relation $a''=\sum F''$ where $a''=F'(r_1,\ldots,r_{N/2})$ that binds $F'$. To summarize, with a single round, the original Sumcheck instance is reduced to two Sumcheck instances of polynomials of $N/2$ variables. 

Actually, this protocol does not work as is since we cannot implement oracle access to these polynomials in our setting. We carefully choose the number of variables eliminated in each round so that the polynomial corresponding $F'$ above has exactly $n$ coefficients. Each coefficient is then mapped to a unique node, giving the network oracle access to the polynomial that fully exploits parallel bandwidth of the network.
The same variable elimination procedure is repeated recursively for $F'$, until we reach single-variable polynomials. The implementation of this idea finally gives~\Cref{theo:naivetriangle}.

One caveat of this approach is that each elimination doubles the number of polynomials to be Sumchecked. It can be overcome by the folding technique of~\cite{levrat2025divide}, which enables us to fold multiple instances of Sumcheck into a single Sumcheck. Crucially, the number of monomials in the intermediate polynomials in the whole procedure does not exceed $n$, so oracle access to these polynomials can be realized by short certificates.

\paragraph*{Round complexity of \nonkcol\ on constant-degree graphs.}
We now turn back to the polynomial $T\cdot S$ for \nonkcol. Each variable of the polynomial appears in $T$ at most as many as $k$ times the degree of the corresponding node. Focusing on constant-degree graphs makes the individual degree of the polynomial $\bigoo(1)$. Here, we use the same idea for triangle counting. In each repetition, we reduce the Sumcheck relation $a=\sum F$ for $N$-variate polynomial into two Sumcheck relations $a'=\sum h$ and $a''= \sum \widetilde{h}$, where $h$ is a $\bigoo(\log n)$-variate polynomial which is the sum of $F$ over the sub-hypercube on all variables except the first $\bigoo(\log n)$ variables. This polynomial can be specified by $\bigoo(n)$ monomials, as the individual degree of $F$ is $\bigoo(1)$. In particular, oracle access to $h$ is achieved by distributing the coefficients.  
Another polynomial $\widetilde{h}$ consists of exponentially many monomials, and we repeat the same procedure for the Sumcheck relation $a''= \sum \widetilde{h}$ to eliminate $\bigoo(\log n)$ variables at once. Eventually, we get $\bigoo(n/\log n)$ instances of Sumcheck each concerning an $\bigoo(\log n)$-variable polynomial. The whole procedure to generate these instances requires $\bigoo(n/\log n)$-round, and $\bigoo(\log^2n)$-bit of messages. To solve these instances in $\bigoo(n/\log n)$ rounds, we process $\bigoo(\log n)$ instances at once in parallel using the standard Sumcheck, as we can use $\bigoo(\log^2n)$-bit of messages. To solve all instances, we need to repeat this for $\bigoo(n/\log^2 n)$ times sequentially, each requiring $\bigoo(\log n)$ rounds.

Finally, to achieve a barrier to lower-round complexity, we leverage the reduction from 3-SAT to 3-coloring. (Note that the direction of reduction is opposite to the reduction used to obtain the polynomial for this problem.) Starting from an instance of 3-SAT, it is first reduced to $2^{o(n)}$ many sparse formulas using the sparsification lemma~\cite{impagliazzo2001problems}. Each sparse instance is converted to an instance of 3-coloring, and due to the sparsity, the converted graph has a constant maximum degree. Now, assuming that each node is a polynomial-time Turing machine, it is not difficult to simulate a distributed interactive proof in the centralized setting using the same number of rounds. Finally, the round reduction technique~\cite{goldreich2002interactive} applied to the simulated interactive proof yields a subexponential-time Arthur-Merlin protocol to solve unsatisfiability of the original instance, which obviously breaks the Arthur-Merlin Exponential-Time Hypothesis.

\subsection{Other related work}\label{subsec:related-work}
    Recently in ~\cite{modanese2025strong, jauregui2026stronghiding} it was defined the notion of hiding locally checkable proofs, which is reminiscent of zero-knowledge proofs. They constructed hiding LCPs of $k$-coloring, that is, LCPs that certify \textsc{$k$-COLORABILITY} while hiding the actual colors in the sense that no constant-round distributed algorithm can recover the $k$-coloring from the provided proofs. This is different from dZK proofs; for $2$-coloring, giving a valid $2$-coloring reveals nothing new, so it is a dZK proof, but is not a hiding LCP.
    
    A series of papers has studied efficient protocols in the $\mathrm{dIP}$ model~\cite{montealegre2021compact,jauregui2022distributed,gil_et_al:LIPIcs.DISC.2025.34,gil2026distributed}. Several trade-off results such as space vs. communication or shared vs. private randomness are explored by~\cite{crescenzi2019trade}. The role of shared randomness is also studied in~\cite{montealegre2025shared}. In~\cite{le2023distributed}, the quantum variant of $\mathrm{dIP}$ is introduced, extending the framework of non-interactive quantum proofs~\cite{fraigniaud2021distributed,le2023distributed-2,hasegawa2024power}. Distributed verification itself was first introduced by Ref.~\cite{korman2010proof}. Its variants~\cite{fraigniaud2019randomized,fraigniaud2019distributed,censor2020approximate,goos2016locally} have been studied as well.

    Zero-knowledge implementations of Sumcheck have been considered in the centralized IPs~\cite{chiesa2017zero,ben2017zero,chiesa2022spatial,gur2024perfect,10.1145/3717823.3718128}. These results achieve perfect zero-knowledge, but focus on other models such as (interactive) PCPs or multi-prover IPs, instead of IPs. In the IP model, while Sumcheck can be made \emph{computational} zero-knowledge assuming the existence of one-way functions~\cite{ben1990everything}, achieving statistical zero-knowledge for Sumcheck would imply complexity-theoretic collapses~\cite{fortnow1987complexity}.

    Distributed algorithms concerning security against adversaries have been considered in several settings~\cite{DBLP:conf/innovations/HitronPY23,DBLP:conf/wdag/HitronP21,DBLP:conf/podc/FischerP23,parter2019distributed,parter2019secure}. Additionally, \cite{aldematshuva_et_al:LIPIcs.DISC.2024.1,10.1007/978-3-031-48615-9_3} studied efficient verification of distributed $\mathsf{CONGEST}$ algorithms with a computationally limited prover.

\section{Preliminaries.}

An undirected graph \(G=(V,E)\) consists of a finite set \(V\) of nodes and a set
\(E \subseteq \binom{V}{2}\) of unordered pairs of distinct nodes, called edges.
Throughout the paper, we let \(n := |V|\).

A \emph{path} in \(G\) is a sequence of nodes \(v_0, v_1, \dots, v_k\) such that
\(\{v_{i-1}, v_i\} \in E\) for all \(i \in \{1,\dots,k\}\).
The \emph{length} of the path is \(k\).
The \emph{distance} between two nodes \(u,v \in V\), denoted \(dist_G(u,v)\),
is the minimum length of a path connecting \(u\) and \(v\), if such a path exists.

The graph \(G\) is \emph{connected} if for every pair of nodes \(u,v \in V\),
there exists a path connecting \(u\) and \(v\). A \emph{subgraph} of \(G\) is a graph \(G'=(V',E')\) such that
\(V' \subseteq V\) and \(E' \subseteq E \cap \binom{V'}{2}\).
If \(V'=V\), then \(G'\) is called a \emph{spanning subgraph} of \(G\). A \emph{tree} is a connected graph with no cycles.

A \emph{spanning tree} of a connected graph \(G\) is a subgraph
\(T=(V,E_T)\) that is a tree and satisfies \(E_T \subseteq E\).

For a node \(v \in V\), we denote by
\(N(v) := \{ u \in V \mid \{u,v\} \in E \}\)
the set of neighbors of \(v\), and by \(\deg(v) := |N(v)|\) its degree.
The \emph{maximum degree} of \(G\) is denoted by
\(\Delta(G) := \max_{v \in V} \deg(v)\).

For an integer $m$, we use $[m]:=\{1,\ldots, m\}$. Given a prime $q$, we denote as $\mathbb{F}_q = \{0,1,\dots,q-1\}$ the finite field of order $q$ (with addition and multiplication taken modulo $q$).  We denote its multiplicative group by $\mathbb{F}_q^{\times} = \mathbb{F}_q \setminus \{0\}.$ For any finite set $\Omega$, we write
$X \sim \mathrm{Unif}(\Omega)$
to mean that $X$ is drawn uniformly at random from $\Omega$, or simply $X\in_R\Omega$.

 Given two measures \(\mu\) and \(\nu\) on a measurable space \((\Omega,\mathcal{F})\), the {\it total variation distance} $\TVc{\mu}{\nu}$ is defined as follows.
\begin{definition}[Total Variation Distance]\label{def:tvdist} Given a measure space $(\Omega, \mathcal F)$ and two measures $\mu, \nu$ on $(\Omega, \mathcal F)$, the total variation distance $\TV{\mu}{\nu}$ between $\mu$ and $\nu$ is defined as 

$$\TV{\mu}{\nu} = \sup_{A\subseteq \Omega}|\mu(A)- \nu(A)|$$

If $\Omega$ is a finite set, then 

$$\TV{\mu}{\nu} = \dfrac{1}{2}\sum_{x\in\Omega}|\mu(\{x\})- \nu(\{x\})|$$

In particular, for two probability mass functions \(P\) and \(Q\) on a finite set then 
$$
\mathrm{TV}(P,Q)~=~\frac12 \sum_{x} \bigl|P(x) - Q(x)\bigr|.
$$
    
\end{definition}

Given two random variables $X$ and $Y$ over a finite set $\Omega$ with probability mass functions $P_X$ and $Q_Y$, respectively, we abuse notation and denote as $\TVc{X}{Y}$ the total variation distance between $P_X$ and $Q_Y$, i.e., $\TVc{P_X}{Q_Y}$. Now we define {\it statistical indistinguishability}.

\begin{definition}\label{def:statind}
    Two random variables $X$ and $Y$ defined over the same set $A$ with size $|A| = q$ are said to be {\it statistically indistinguishable} iff for every polynomial $p$
    $$\TVc{X}{Y} = \sum_{z\in A}\abs{\Prob{X=z}-\Prob{Y=z}}\leq 1/p(n)$$
    where $n$ is a given parameter\footnote{In this work, $n$ is always a number of nodes in a distributed network.}.
\end{definition}

The following are known and standard properties of the total variation distance.

\begin{property}[\cite{cover-thomas}]\label{lem:product-independent}
    For three random variables $\mathsf X,\mathsf Y,\mathsf Z$ on $\mathbb F_q$, where $\mathsf Z$ is independent from $\mathsf X$ and $\mathsf Y$,
    \[
    \TVc{\mathsf X\cdot \mathsf Z}{ \mathsf Y\cdot \mathsf Z} \leq \TVc{\mathsf X}{\mathsf Y}
    \]
\end{property}

\begin{property}[\cite{tao-vu}]\label{lem:product-uniforms}
    Let $\mathsf X_1,\ldots, \mathsf X_h\overset{\textit{i.i.d.}}{\sim} \Unif(\mathbb F_q)$. Then,
    \[
    \TVc{\prod_i \mathsf X_i}{\Unif(\mathbb F_q)} = O\left(\frac hq \right).
    \]
\end{property}

\begin{property}[\cite{durrett}]\label{lem:partition} Let $X,Y$ be two random variables on $\fieldq$ and $\{\mathcal{E}_i\}_{i\in[k]}$ a partition of events. Then 

\[ \TVc{\mathsf X}{\mathsf Y} = \sum_{i=1}^k \Prob{\mathcal{E}_i}\cdot\TVc{\mathsf X|\mathcal{E}_i}{\mathsf Y|\mathcal{E}_i}\]
    
\end{property}

\subsection{Distributed interactive proofs and distributed zero-knowledge proofs}\label{subsec:dipandzk}
A distributed language $\mathcal L$ is a set of pairs $(G,I)$ where 
\begin{itemize}
    \item $G=(V,E)$ is a connected communication graph,
    \item $I:V\rightarrow \{0,1\}^*$ is a function such that $I(v)$ represents the input to $v$ (e.g., the weights of edges incident to $v$, or the identifier of $v$). In this paper, the input label $I(v)$ always contains the unique identifier of $v$ $\mathsf{ID}(v)\in \{0,1,\ldots, n-1\}$, and the ID of its parent in an arbitrary spanning tree rooted a the node with ID $0$.
\end{itemize}
For integer $k\in \{0,\ldots ,n-1\}$ the node $k$ means the node $v$ with $\mathsf{ID}(v)=k$.

Our verification model for a distributed language follows the \emph{distributed
interactive proof} (dIP) framework introduced
by Kol, Oshman, and Saxena~\cite{kol2018interactive}.  Informally, the $n$ nodes of a connected communication network $G=(V,E)$ act as a (distributed) verifier who may interact with a single powerful but untrusted prover. A distributed interactive proof is parametrized by 
\begin{itemize}
    \item the number of rounds $r$ of interactions between nodes and the prover.
    \item the number of bits $\ell$ in the messages exchanged between each node and the prover.
    \item the class $\mathcal A_V$ of distributed verification algorithms.
\end{itemize} 
Therefore we refer to it as a $(r,\ell,\mathcal A_V)$-distributed interactive proof, or simply $(r,\ell,\mathcal A_V)$-dIP. An $(r,\ell,\mathcal A_V)$-dIP proceeds as follows.
First, the network and the prover interact in $r$ rounds, such that the final round is the prover's round. In each round, the prover sends each node a message of $\ell$-bit (if it is the prover's round) or each node sends the prover a message of $\ell$-bit (if it is the verifier's round). After these interaction rounds, the network, without further communication with the prover, runs a verification algorithm $A\in \mathcal A_V$ which decides the output of each node. 

\begin{definition}
 The class $\mathrm{dIP}[r,\ell,\mathcal{A}_V]$ is defined as all the distributed languages $\mathcal L$ such that there exists an $(r,\ell,\mathcal{A}_V)$-distributed interactive proof for $\mathcal L$ satisfying the following conditions:
\begin{description}
    \item[Completeness:] If $(G,I)\in \mathcal L$, there exists a prover that makes all nodes output ``accept" with probability at least $2/3$.
    \item[Soundness:] If $(G,I)\notin \mathcal L$, for any prover, all nodes output ``accept" with probability at most $1/3$.
\end{description}

Throughout this work, the class $\mathcal A_V$ of verification algorithms will always be a class of distributed algorithms $\mathcal C[\bigoo(1), r\ell]$ that run in $\bigoo(1)$ synchronous rounds, and in each round the neighbors exchange messages of size at most $\bigoo(r\ell)$ bits. Note that in the original definition~\cite{kol2018interactive} the class was further limited to one-round algorithms.

\end{definition}

 \subsubsection{Zero-knowledge property}\label{subsec:high-levelzk}

In this work, we study distributed zero-knowledge interactive proofs against {\it an honest but curious node}. This means that each node will not deviate from the prescribed protocol, but cannot learn anything from the information provided by the prover and the rest of the nodes.

Motivated by the definition of {\it perfect} zero-knowledge in the centralized setting, \cite{zkdef} introduced the notion of distributed zero-knowledge proofs. Following \cite{zkdef,grilo2025distributed}, the {\it view} of a node $v$ in a distributed interactive proof is defined as follows.

\begin{definition}[View of a node $v$, \cite{zkdef}] Given a graph $(G,I)$ and a dIP $\Pi$, the
view of a node $v\in V$ is defined by a random variable $\textsc{VIEW}(\Pi,G,I)_v$ that contains all the values given to the node $v$, either by the prover or by its neighbors during the protocol $\Pi$.  $\mathcal{VIEW}_v$ denotes the distribution of \textsc{VIEW}$_v$.

\end{definition}

In order to formally define the notion of distributed zero-knowledge, it is necessary to consider an additional class of distributed algorithms $\mathcal A_S$, that is, a class of simulators. (However, throughout this work, the class is always identical to $\mathcal A_V$.) As in the original definition of dZK given in \cite{zkdef}, the notion of zero-knowledge can be extended for a coalition of nodes, but we focus on the most fundamental case: zero-knowledge against coalitions of size one.
Therefore the definition of zero-knowledge to be used in this work is the following.

\begin{definition}[dZK for coalitions of size one, \cite{zkdef} and \cite{grilo2025distributed}]\label{def:perfectzk}
    Let $r, \ell \in \mathbb{N}$ and let $\mathcal{A}_V, \mathcal{A}_S$ be non-empty sets of distributed algorithms.
The class $\textnormal{dZK}[r, \ell, \mathcal{A}_V, \mathcal{A}_S]$ is the set of all distributed languages $\mathcal L$, for which there exist an $(r, \ell, \mathcal{A}_V)$-distributed interactive proof system 
such that there exists a simulator $S\in \mathcal A_s$ such that for every node $v\in V$, it holds that
$$\left(\mathcal{OUT_S}(G,I)\right)_{v} \equiv \left(\mathcal{VIEW}(\Pi,G,I)\right)_{v}$$
where the equality holds in distribution, $\left(\mathcal{OUT_S}(G,I)\right)_{v}$ corresponds to the distribution of the view of node $v$ in the simulator. 
\end{definition}

This definition of zero-knowledge given in \cite{zkdef} is a distributed adaptation of the notion of {\it perfect} zero-knowledge in the centralized interactive proofs, in which we ask for the existence of a simulator such that the view of each node after the execution of the simulator has the same distribution as the view obtained during a real execution of the interactive proof. Following the definition of statistical zero-knowledge proofs in the centralized setting~\cite{goldwasser1989knowledge}, in which we now ask for the existence of a simulator such that the view of each node after the execution of the simulator is {\it statistically close} to the distribution of its view in the interactive proof real distribution. The notion of {\it distance} between different distributions is the total variation distance (see \Cref{def:tvdist}). , we define the class of {\it statistical} zero-knowledge distributed proofs, which we call dStatZK, defined as follows

\begin{definition}[dStatZK]\label{def:statzk} Let $r, \ell \in \mathbb{N}$ and let $\mathcal{A}_V, \mathcal{A}_S$ be non-empty sets of distributed algorithms.
The class of distributed {\it statistical} zero-knowledge proofs $\textnormal{dStatZK}[r, \ell, \mathcal{A}_V, \mathcal{A}_S]$ is defined as the set of all distributed languages $\mathcal L$, for which there exist an $(r, \ell, \mathcal{A}_V)$-distributed interactive proof system  
and a simulator $S \in \mathcal{A}_S$ that satisfy the following {\it statistical zero-knowledge property:} The two distributions $\left(\mathcal{OUT_S}(G,I)\right)_{v}$ and $\left(\mathcal{VIEW}(\Pi,G,I)\right)_{v}$ are statistically indistinguishable (\Cref{def:statind}).

In the special case of $\mathcal A_V = \mathcal A_S$, we refer to the class $\textnormal{dStatZK}[r, \ell, \mathcal{A}_V, \mathcal{A}_S]$ simply as $\textnormal{dStatZK}[r, \ell, \mathcal{A}_V]$.
\end{definition}

\section{Sumcheck Protocol: General Framework.}\label{sec:distsc}

Assume that we are given oracle access to a polynomial $F:\mathbb F^N\rightarrow \mathbb F$ of total degree $d$ and individual degree $\bigoo(n)$ and $N$ variables.
The Sumcheck protocol, originally due to Lund, Fortnow, Karloff, and Nisan~\cite{lund1992algebraic}, verifies the value of a large summation of $F$ over the Boolean hypercube $\{0,1\}^N$, converting a claim about the sum of $2^N$ polynomial evaluations to a single evaluation at a randomly chosen point. Below we briefly describe the protocol, which is the same as the one presented in~\Cref{subsec:overview}.

\begin{algorithmbox}{\textsc{Sumcheck-Protocol}}{alg:sumcheck}
{
\textbf{Input}: A polynomial $F:\fieldq^N\rightarrow \fieldq$ of total degree $d$, and a field element $a\in \mathbb F_q$.\\
\textbf{Goal}: Check that $\sum_{\mathbf{x}\in\{0,1\}^N}F(\mathbf{x})=a$.
   \begin{enumerate}
\item \textbf{Step 1.}
\begin{itemize}
\item The prover sends an univariate polynomial \( g_1(x_1) \). 
\item The verifier checks $g_1(0) + g_1(1) = a$, rejects otherwise.
\end{itemize}

\item \textbf{Step 2.} Repeat the following for \( 2 \leq i \leq N \):
\begin{itemize}
\item The verifier picks \( r_{i-1} \in \mathbb{F} \) u.a.r., sends to the prover.
\item The prover sends an univariate polynomial \( g_i(x_i) \).
\item The verifier checks \( g_i(0) + g_i(1) = g_{i-1}(r_{i-1}) \), rejects otherwise.
\end{itemize}

\item \textbf{Final Check}:
The verifier picks \( r_{N} \in_R \fieldq \), checks \( g_N(r_N) = F(r_1,\ldots,r_N) \), rejects otherwise.
\end{enumerate}
}
\end{algorithmbox}

\begin{theorem}[Completeness and Soundness, \cite{lund1992algebraic}]\label{thm:completeness-soundness-sumcheck}
The Sumcheck protocol satisfies:
\begin{itemize}
\item \textbf{Completeness}: If $\sum_{\mathbf{x}\in\{0,1\}^N}F(\mathbf{x}) = a$, the honest prover can convince the verifier with probability 1 by sending \( g_i(x_i) = \sum_{x_{i+1},\ldots,x_n \in \{0,1\}} F(r_1,\ldots,r_{i-1},x_i,\ldots,x_n) \).
\item \textbf{Soundness}: If $\sum_{\mathbf{x}\in\{0,1\}^N}F(\mathbf{x}) \neq a$, for any (cheating) prover the verifier rejects with probability at least \(1 - \frac{N\cdot d}{q} \).
\end{itemize}
\end{theorem}
The proof follows from the standard polynomial identity test: at the first round in which $g_i$ is not the correct partial sum, the polynomial sent differs from the honest one by a non‑zero degree‑$d$ polynomial, which vanishes at most $d$ points due to Schwartz-Zippel Lemma. The random check catches this with probability $1 - d/q$.  An union bound over the $N$ rounds yields the claimed error bound. See~\cite{lund1992algebraic} for more detail.

\subsection{The Distributed Implementation}\label{subsubsec:scimpl}
We address how to solve a general instance of the distributed Sumcheck protocol. Formally, the problem that we address is the following.

\begin{tcolorbox}[title=\sumcheckproblem]
    \textbf{Input.} A Sumcheck-instance $(G,F,q,a)$, where $G$ is an $n$-node communication graph of the distributed verifier,  $F\colon \mathbb F_q^N\to \mathbb F_q$ is a polynomial of individual degree at most $\bigoo(n)$, $q$ is a prime number representing the size of $\fieldq$ and a value $a \in \mathbb F_q$.\\
    
    \textbf{Assumption.} An Oracle access to $F$ and an arbitrary spanning tree $T$ of $G$. \\
    
    \textbf{Decision problem:} Decide if $\sum_{x\in\{0,1\}^N} F(x) = a$ or not: If $\sum_{x\in\{0,1\}^N} F(x) = a$, all nodes accept; otherwise, at least one node rejects.
\end{tcolorbox}

In this problem, we just assume oracle access to $F$ and ignore how to realize it. Later, when the applications of this framework are explained, we will address how to evaluate $F$ for each concrete problem.

As a warm-up, let us first discuss how to simulate each round of Protocol~\ref{alg:sumcheck} using a distributed interactive proof without zero-knowledge guarantees. Here, the prover is required to send an univariate polynomial $g_i$, which is allegedly equal to  
$$\underset{x_{i+1}\in\{0,1\},\ldots,x_N\in\{0,1\}}{\sum} F (r_1,\ldots, r_{i-1},x,x_{i+1},\ldots,x_N).$$ 
Since the individual degree of $F$ is $\bigo{n}$, this polynomial has degree at most $\bigo{n}$ (for simplicity, we assume that the degree is less than $n$). Therefore alleged $g_i$ must be specified by $n$ field elements that represent coefficients: $g_i(x)=\sum_{j=0}^{n-1} \alpha_j^{(i)}x^j$. In the protocol, the prover distributes these coefficients to the nodes, so that each node $k$ receives a single field element $\alpha_k^{(i)}$. In case of $\bigo{n}$-degree, each node receives $\bigo{1}$ coefficients.

For the initially given spanning tree $T$, let $\text{child}(k)$ be the set of children of $k$ in $T$ and $T_k$ be the subtree of $T$ with the root $k$. We now explain the verification phase of the protocol. 
\begin{description}
    \item[Simulating Step 1.] Step 1 of Sumcheck protocol checks that \( g_1(0) + g_1(1) = a\). Since $g_1(x) = \sum_k \alpha_k^{(1)}x^k$, it is equivalent to checking \( 2\alpha_0^{(1)} + \sum_{k=0}^{n-1} \alpha_k^{(1)} = a\). To this end, we introduce a value $\beta_k^{(1)}$ provided by the prover to node $k$, which is allegedly equal to $\sum_{j\in T_k} \alpha_j^{(1)}$. The correctness of this value is checked as follows, by 1-round communication. 
    
    \begin{itemize}
        \item If $\text{child}(k)=\emptyset$ ($k$ is a leaf in $T$), $k$ checks that
        \begin{equation}\label{eq:s11}
        \alpha_k^{(1)}= \beta_k^{{(1)}}
        \end{equation}
        \item If $k$ is not a leaf, checks that 

        \begin{equation}\label{eq:s12}
            \beta_k^{(1)}= \alpha_k^{(1)}+ \sum_{j \in \text{child}(k)}\beta_j^{(1)}
        \end{equation}
        where $\beta_j^{(1)}$ is sent from $j$.

        \item If $k$ is the root of the tree, i.e., $k=0$, it also checks that 
        \begin{equation}\label{eq:s13}
            \alpha_0^{(1)} + \beta_0^{(1)}=a.
        \end{equation}
    \end{itemize}
    
     If all these tests passed, then $g_1(0) + g_1(1) = a$.
    \item[Simulating Step 2.] Fix arbitrary $i$. Step 2 of Sumcheck protocol checks that \( g_i(0) + g_i(1) = g_{i-1}(r_{i-1}) \). This is equivalent to checking
    \[
    \alpha_0^{(i)} + \sum_{k=0}^{n-1} \alpha_k^{(i)} = \sum_{k=0}^{n-1} \alpha_k^{(i-1)}r_{i-1}^k.
    \]
    To this end, we introduce $\beta_k^{(i)}$ and $\widetilde{\beta_k^{(i)}}$ provided by the prover to node $k$, allegedly satisfying:
    \begin{align*}
        \beta_k^{(i)} &= \sum_{k'\in T_k} \alpha_{k'}^{(i)}\\
        \widetilde{\beta_k^{(i)}} &= \sum_{k'\in T_k} \alpha_{k'}^{(i-1)} r_{i-1}^{k'}.
    \end{align*}
    Similarly to the simulation of Step 1, the correctness of these values can be checked using the tree $T$: If $k$ is a leaf, checks that $\alpha_k^{(i)}= \beta_k^{(i)}$ and $\widetilde{\beta_k^{(i)}}=\alpha_k^{(i)}\cdot r_{i-1}^k$.
    If $k$ is not a leaf, checks that 
        \begin{equation*}
            \beta_k^{(i)}= \alpha_k^{(i)}+ \sum_{j \in \text{child}(k)}\beta_j^{(i)}
        \end{equation*}
        and
        \begin{equation*}
            \widetilde{\beta_k^{(i)}}= \alpha_k^{(i-1)}\cdot r_{i-1}^k+ \sum_{j \in \text{child}(k)}\widetilde{\beta_j^{(i)}}
        \end{equation*}
        where $\beta_j^{(i)}$ and $\widetilde{\beta_j^{(i)}}$ are sent from $j$.
    Then node $k=0$ additionally checks that $\alpha_0^{(i)}+\beta_0^{(i)} = \widetilde{\beta_0^{(i)}}$.
    \item[Simulating Final Check.] Final Check of Sumcheck protocol verifies that indeed the equality \( g_N(r_N) = F(r_1,\ldots,r_N) \) holds. As in the simulation of Step 1 and Step 2, $g_N(r_N)$ can be computed using the tree $T$, and oracle access to $F$ (from our assumption) gives the value of $F(r_1,\ldots,r_N)$.
\end{description}

The distributed implementation of the Sumcheck protocol described above follows the same sequence of steps as the centralized algorithm in Protocol~\ref{alg:sumcheck}, but propagates intermediate data along a rooted spanning tree of the communication graph.
Therefore, it is not zero-knowledge because, through intermediate computations, the verifier learns the coefficients of internal polynomials and the evaluations of each internal polynomial $g_i$ for random values chosen by the verifier. Specifically, the root node receives the evaluation of $g_i$. 
Since the original polynomial encodes global information, learning these values leaks information of the graph, which is impossible to learn for the class of simulators $\mathcal C[\bigo{1},\bigo{N\log q}]$, which we regard as {\it efficient}. 

To resolve this issue, we show how to simulate each step of the distributed implementation described above, keeping these values secret from the nodes.
   
\subsection{General Toolbox for the zero-knowledge simulation.}\label{subsec:general-toolbox}

Before describing the formal modification of Sumcheck, we describe the main technical modifications that will be done to the original algorithm in order to hide the coefficients and partial sums.

Regarding the implementation of Steps 1 and 2 in Sumcheck Protocol, in both cases the original implementation must ensure that each node receives the appropriate coefficients of the univariate polynomials defined by:

$$
\sum_{x_{i+1}\in\{0,1\},\dots,x_N\in\{0,1\}} f(r_1,\dots,r_{i-1},x,x_{i+1},\dots,x_N).
$$
As previously discussed, revealing the coefficients of these polynomials may inadvertently leak information that cannot be simulated by any algorithm in the class $\mathcal{C}[\bigo{1}, \bigo{\log n}]$. Nonetheless, it is still necessary to verify that certain conditions hold for these polynomials such as Equations \ref{eq:s11}, \ref{eq:s12}, and \ref{eq:s13}.
The technique used to encrypt the sum of coefficients is fairly standard in the context of secure multiparty computation (MPC) \cite{shamir, rabin1989verifiable}, and has been employed in previous distributed zero-knowledge implementations as well (e.g., for 3-coloring in \cite{zkdef}). However, in our setting, a key difference that must be addressed is that each value sent to a node may participate in more than one equation. This leads to a more intricate construction, which we develop incrementally, culminating in the complete protocol described in \Cref{subsec:final-protocol}. A central tool is the use of a {\it polynomial encryption} of a value $\alpha \in \fieldq$.

\begin{definition}\label{def:polyenc}
    Given a field $\fieldq$ with $q$ prime and a value $\alpha\in \fieldq$, a polynomial encryption of $\alpha$ is a degree-one polynomial, denoted as $P[\alpha]\in\fieldq[x]$, such that $P[\alpha]$ is built by sampling an uniform value $r\overset{}{\sim}\Unif(\fieldq)$ and 

    $$P[\alpha](x) = r\cdot x + \alpha$$
\end{definition}

The following standard fact about interpolation of a degree-one polynomial holds.

\begin{fact}\label{fact:interpol}
    Given a field $\fieldq$ and two different points $y_1,y_2\in\fieldq$, there exists a unique polynomial $P(x)\in\fieldq^1[x]$, where $\fieldq^1[x]$ is the set of all degree-one polynomials over $\fieldq$, such that both $y_1$ and $y_2$ are in the image of $P(x)$ in $\fieldq$.
\end{fact}
Fact \ref{fact:interpol} states that knowing two evaluations of the polynomial enables us to reconstruct such a polynomial due to its uniqueness (for example, through interpolation). In fact, the following property is crucial to obtain a constant round simulator for both Steps.

\begin{proposition}\label{lemma:unirandvalues}
    Given $\alpha\in \fieldq$, $P[\alpha]$ its polynomial encryption and an polynomial $N\in_R\fieldq^1[x]$ chosen uniform at random, then for any $\beta\in\fieldq$, the distribution of the random variable $\delta\colon=P[\alpha](\beta)+N(\beta)$ is uniform in $\fieldq$.
\end{proposition}

    \begin{proof}
    Fix $\beta\in\fieldq$. Since $P[\alpha]$ is fixed once $\alpha$ is fixed, the value
    $P[\alpha](\beta)$ is a fixed element of $\fieldq$. Hence it is enough to prove that
    $N(\beta)$ is uniformly distributed over $\fieldq$.

    Write $N(x)=ax+b$, where $a,b\in\fieldq$ are chosen uniformly at random subject to
    $N\in_R\fieldq^1[x]$. For any fixed value $z\in\fieldq$, the condition
    $N(\beta)=z$ is equivalent to
    $a\beta+b=z.$
    For every admissible choice of $a$, there is a unique value of $b$, namely
    $b=z-a\beta$, satisfying this equation. Therefore, the number of polynomials
    $N\in\fieldq^1[x]$ such that $N(\beta)=z$ is the same for every $z\in\fieldq$.
    Consequently, $N(\beta)$ is uniformly distributed over $\fieldq$.

    Finally, $\delta=P[\alpha](\beta)+N(\beta)$ is obtained by adding the fixed field
    element $P[\alpha](\beta)$ to a uniformly distributed random variable. Since
    translation by a fixed element is a bijection of $\fieldq$, $\delta$ is also
    uniformly distributed over $\fieldq$.
\end{proof}
\subsubsection{Zero-knowledge implementation of Step 1.} 
Without loss of generality, we assume that the sum to be verified is $0$. Recall in Step 1 each node $k\in V$ receives two values $\alpha_k^{(1)},\beta_{k}^{(1)}\in \fieldq$ where $g_1(x)=\sum_{k=0}^{n-1}\alpha_k^{(1)}x^k$ and $\beta_k^{(1)}= \sum_{j\in T_k} \alpha_j^{(1)}$ for the subtree $T_k$ with root $k$, and in the verification phase each node $k$ needs to verify 

\begin{equation}\label{eq:111}
    \beta_k^{(1)} =\alpha_k^{(1)} + \sum_{j \in \text{child}(k)}\beta_j^{(1)}
\end{equation}
and node $k=0$ also needs to verify that 

\begin{equation}\label{eq:112}
    \alpha_0^{(1)} +\beta_0^{(1)} = 0.
\end{equation}

 Given the complexity of these technical aspects, we first detail how Step 1 can be implemented using a protocol satisfying the perfect zero-knowledge property, and subsequently describe how to extend this approach to handle Step 2 while preserving the same guarantees.

\paragraph{Random Additions.} \label{subsubsec:random-addition}
In order to conceal the values $\alpha_k^{(1)}$ and $\beta_k^{(1)}$ from the nodes, the prover computes a random polynomial encryption, denoted by $P\left[\alpha_k^{(1)}\right]$ and $P\left[\beta_k^{(1)}\right]$. Additionally, for each of these values, the prover independently samples random polynomials of degree one, $N\left[\alpha_k^{(1)}\right](x), N\left[\beta_k^{(1)}\right](x) \in_R \fieldq^1[x]$ for node $k$. The coefficients of these random polynomials do not depend on $\alpha_k^{(1)}$ or $\beta_k^{(1)}$. The only constraint imposed on the resulting family of polynomials is the following condition:

\begin{align}\label{eq:eqcondition}
   \textbf{(RC)}_1\quad &N\left[\alpha_k^{(1)}\right](0)  + \sum_{j \in \text{child}(k)} N\left[\beta_j^{(1)}\right](0) = N\left[\beta_k^{(1)}\right](0)\quad \forall k \in V 
\end{align}

After the prover computes these polynomials, the simulation of Step 1 is done by checking $W_k(0) = 0$ where
\begin{equation}\label{eq:30}
    W_k(x) = \left(P\left[\alpha_k^{(1)}\right]  + N\left[\alpha_k^{(1)}\right]\right)(x)+ \sum_{j \in \text{child}(k)} \left(P\left[\beta_j^{(1)}\right] + N\left[\beta_j^{(1)}\right]\right)(x) - \left(P\left[\beta_k^{(1)}\right]  + N\left[\beta_k^{(1)}\right]\right)(x)
\end{equation}
since the condition \textbf{(RC)}$_1$ gives
\begin{equation}\label{eq:9}
    W_k(0) = 0\iff \alpha_k^{(1)} + \sum_{j\in \text{child}(k)} \beta_j^{(1)} = \beta_k^{(1)}.
\end{equation}

Each node $k$ needs to reconstruct $W_k(x)$ without knowing polynomial encryptions. This is achieved as follows. The prover sends to each node $k$ the following values, where $c_k\in \{1,2\}$ is the color of $k$ in a random $2$-coloring of $T$:

\begin{itemize}
    \item $P\left[\alpha_k^{(1)}\right](c_k)$, $N\left[\alpha_k^{(1)}\right](c_k)$, $P\left[\beta_k^{(1)}\right](c_k)$, $N\left[\beta_k^{(1)}\right](c_k)$ 
    \item $P\left[\alpha_p^{(1)}\right](c_k)$, $N\left[\alpha_p^{(1)}\right](c_k)$, $P\left[\beta_p^{(1)}\right](c_k)$, $N\left[\beta_p^{(1)}\right](c_k)$ of its parent $p$ in the spanning tree $T$
    \item $P\left[\beta_j^{(1)}\right](c_k)$, $N\left[\beta_j^{(1)}\right](c_k)$ of each child $j$ in the spanning tree $T$.
\end{itemize}
Once the prover sends them, the node $k$ sends its parent $p$
\begin{itemize}
    \item $P\left[\alpha_k^{(1)}\right](c_k) + N\left[\alpha_k^{(1)}\right](c_k)$, $P\left[\beta_k^{(1)}\right](c_k)+N\left[\beta_k^{(1)}\right](c_k)$ 
    \item $P\left[\alpha_p^{(1)}\right](c_k)+N\left[\alpha_p^{(1)}\right](c_k)$, $P\left[\beta_p^{(1)}\right](c_k)+N\left[\beta_p^{(1)}\right](c_k)$
\end{itemize}
to reconstruct the polynomials $P\left[\alpha_k^{(1)}\right] + N\left[\alpha_k^{(1)}\right], P\left[\beta_k^{(1)}\right] + N\left[\beta_k^{(1)}\right]$,  $P[\alpha_j^{(1)}] + N\left[\alpha_j^{(1)}\right]$ for each children $j$. Now, the node $k$ can locally compute $W_k(x)$. This strategy allows nodes to verify whether Equation~\eqref{eq:111} holds without revealing the labels $\alpha_k^{(1)}$ and $\beta_k^{(1)}$. 

Nevertheless, this approach introduces three immediate drawbacks that must be addressed.

\begin{description}
    \item[{\bf Problem 1.}] The nodes cannot reconstruct the random polynomials sent by the prover, and therefore cannot trust that the prover in fact sent them random polynomials satisfying condition \textbf{(RC)}$_1$.
    
    \item[{\bf Problem 2.}] As each node can have up to $\Omega(n)$ children, the message size is $\bigo{n\log q}$ bits.
    
    \item[\bf Problem 3.] The nodes were able to verify that Equation \ref{eq:111} is satisfied without revealing any of the coefficients, but we still need to verify that Equation \ref{eq:112} is satisfied, without having access to the values $\alpha_0^{(1)}$ and $\beta_0^{(1)}$.
\end{description}
Problems 1 and 2 are standard and their solutions are included for the sake of completeness. We specifically address Problem 3. These techniques will appear again in the analysis of Step 2.

\paragraph{Cut and Choose Technique.}\label{subsubsec:cutandchoose}

    To resolve {\bf Problem 1.}, we use the cut-and-choose technique, a classical tool used in zero-knowledge algorithms, which is also used in \cite{zkdef}. This technique works as follows: Given a constant parameter $t\in\mathbb N$, the prover computes, for each node $k\in V$, $t$ independent random copies of polynomials $\pal{N_h}{k}{{(1)}}$ and $\pbet{N_h}{k}{{(1)}}$ satisfying condition {\bf (RC)}$_1$ for each $h\in [t]$. Then the prover sends, in the first round of communication, to each node $k$ a share $\pal{N_h}{k}{(1)}(c_k)$ and $\pbet{N_h}{k}{(1)}(c_k)$ for each $h\in [t]$ and sends the shares $\pal{N_h}{k}{(1)}(3- c_k)$ and $\pbet{N_h}{k}{(1)}(3- c_k)$ to each children $j$ of $k$. After the prover has committed the shares of these $t$ copies, in round two the nodes select randomly one of these copies $h^{(1)}\in[t]$ and send this value to the prover. Finally, in round three the prover sends the original messages described before, using the random polynomials $\pal{N_{h^{(1)}}}{k}{(1)}$ and $\pbet{N_{h^{(1)}}}{k}{(1)}$ for each node $k$. For all the copies not selected, the nodes reconstruct the random polynomials by exchanging the shares with neighbors and verify that condition {\bf (RC)}$_1$ holds for each of these copies. The nodes rejects if some of them breaks condition {\bf (RC)}$_1$. The copy $h^{(1)}$ selected is not revealed, and the nodes proceed with this copy to verify the respective equalities. As the prover didn't know which of the $t$ copies were to be used, the probability that the prover sent random polynomials not satisfying \Cref{eq:eqcondition} for the selected copy is at most $1/t$.

    In the zero-knowledge protocol for {\bf Step 2}, we also use extra random polynomials that must satisfy certain equations such as condition {\bf (RC)}$_1$.  Every time we send these polynomials, the protocol takes $\bigo{1}$ rounds and the message grows by a factor of \(t\).  Since \(t\) is a fixed constant, it does not affect the message size \(\bigo{\log n}\).

\paragraph{Messages of size $\bigo{\log n}$}\label{subsec:sizeofmesssages}

To resolve {\bf Problem 2.}, we proceed analogously to \cite{zkdef}. Instead of sending to node $k$ all the shares $\left\{\left[\pbet{P}{j}{(1)}+\pbet{N}{j}{(1)}\right](c_k)\right\}_{j\in \text{child}(k)}$, these shares are distributed among the children of $k$, and later the children are going to inform node $k$ of these shares as follows.

Assume that the children of $k$ are ordered from $j_0,\dots, j_{|\text{child}(k)|-1}$ according to its \id\ (for example an increasing ordering of the \id's). Each child $j_i$ will receive from the prover the share $\pbet{P}{j_{i+1}}{(1)}(c_k)$, $\pbet{N}{j_{i+1}}{(1)}(c_k)$ of the next child $j_{i+1}$. In the verification phase, the node $j_i$ informs to node $k$ the share $\left[\pbet{P}{j_{i+1}}{(1)}+\pbet{N}{j_{i+1}}{(1)}\right](c_k)$ of its consecutive child. 

\paragraph{Same coefficient in different equalities.}
Finally, we deal with {\bf Problem 3.} Using the above techniques the node $0$ can verify $W_0(0)=0$, which means 
$$\alpha_0^{(1)} + \sum_{j\in child(k)} \beta_j^{(1)} = \beta_0^{(1)}$$
by Equation \ref{eq:9}.
Notice that at this point, the node $0$ knows the polynomials $\pal{P}{0}{(1)}+\pal{N_{h^{(1)}}}{0}{(1)}$ and $\pbet{P}{0}{(1)}+\pbet{N_{h^{(1)}}}{0}{(1)}$ and does not know the coefficient $\alpha_0^{(1)}$ nor the value $\beta_0^{(1)}$, but still needs to verify ~\Cref{eq:s13} which is equivalent to $g_1(0) +g_1(1) = 0$. 

The idea is to use extra random polynomials. The algorithm proceeds as follows. In the same round that the prover sent the shares of the $t$ random polynomials, the prover also sends the shares of $t$ extra random polynomials $M_1,\dots,M_t$ to the node $0$ and its children, such that the following condition holds.

\begin{equation*}
    {\bf (0C)}_1\quad M_h(0) = \pal{N_h}{0}{(1)}(0) + \pbet{N_h}{0}{(1)}(0),\quad \forall h\in [t]
\end{equation*}
Consider that the node $0$ reveals $M_h(x)$ for all $h$ and checks if the condition {\bf (0C)}$_1$ holds for all $h\neq h^{(1)}$. Since the node $0$ can recover
$\pal{P}{0}{(1)}+ \pal{N_{h^{(1)}}}{0}{(1)}$ and $\pbet{P}{0}{(1)}+ \pbet{N_{h^{(1)}}}{0}{(1)}$, 
$$\left[\pal{P}{0}{(1)}+ \pal{N_{h^{(1)}}}{0}{(1)} + \pbet{P}{0}{(1)}+ \pbet{N_{h^{(1)}}}{0}{(1)}- {M_{h^{(1)}}}\right](0) = 0$$
can be checked locally. This is identical to $\pal{P}{0}{(1)}(0)+\pbet{P}{0}{(1)}(0)=0$ under the condition that {\bf (0C)}$_1$ holds for $h =h^{(1)}$. Note that this modification additionally reveals 
$\left[\pal{N_{h^{(1)}}}{0}{(1)} + \pbet{N_{h^{(1)}}}{0}{(1)}\right](0)$, resulting that the node $0$ learns the value of $\alpha_0^{(1)}+\beta_0^{(1)}$ (but does not learn $\alpha_0^{(1)}$ and $\beta_0^{(1)}$). Nevertheless, for any yes-instance this value is always equal to $0$ (more precisely, $a$), and the zero-knowledge property survives.

\subsection{Zero-knowledge implementation of Step 2.}

We follow the similar approach for Step 1. Remember that each node $k\in V$ needs to verify the following equalities
\begin{equation}\label{eq:a}
     \beta_k^{(i)} = \alpha_k^{(i)} + \sum_{j\in ch(k)} \beta_j^{(i)}
\end{equation}
\begin{equation}\label{eq:b}
    \widetilde{\beta_k^{(i)}} = \alpha_k^{(i-1)}r_{i-1}^{k}+\sum_{j\in ch(k)} \widetilde{\beta_{j}^{(i-1)}},
\end{equation}
and node $k=0$ also needs to verify that

\begin{equation}\label{eq:c}
        \alpha_0^{(i)}+ \beta_0^{(i)} = \widetilde{\beta_0^{(i)}}.
\end{equation}
The main difference is that we have to use the coefficients from the {\it previous} round to check
\Cref{eq:b}.
To this end, we further modify the cut-and-choose technique. We assume that in the phase $i-1$ the prover actually sent $t^2$ copies of each random polynomial $\pal{N_h}{k}{(i-1)}$, $\pbet{N_h}{k}{(i-1)}$, instead of $t$ copies. The verifier then selected $t^2-t-1$ indices to be opened to check the consistency condition. Among the remaining $t+1$ copies, the verifier selects one index $h^{(i-1)}$ to check the condition involving $\alpha_k^{(i-1)}$, namely, \Cref{eq:a}. Consequently, we have $t$ possible options of the form $\left[ P\left[\alpha_k^{(i-1)}\right] + N_h\left[\alpha_k^{(i-1)}\right] \right] (x)$ for $t$ different indices from $[t^2]$, relabel them by $[t]$. Now in the phase $i$, given the random value $r_{i-1}\in \fieldq$, random polynomials $N_h\left[\widetilde{\beta_k^{(i)}}\right]$ for hiding $\widetilde{\beta_k^{(i)}}$ satisfy
\begin{equation}
    N_h\left[\alpha_k^{(i-1)}\right](0) \cdot r_{i-1}^{k}+\sum_{j\in \text{child}(k)} N_h\left[\widetilde{\beta_{j}^{(i)}}\right](0) - N_h\left[\widetilde{\beta_k^{(i)}}\right](0) = 0.
\end{equation}
Once the above condition is ensured, checking \Cref{eq:b} is done by checking
\begin{equation}\label{eq:step2-2}
    \left[P[\widetilde{\beta_k^{(i)}}] + N_{h}[\widetilde{\beta_k^{(i)}}]\right](0) = \left[P[\alpha_k^{(i-1)}] + N_{h}[\alpha_k^{(i-1)}] \right](0) \cdot r_{i-1}^{k}+\sum_{j\in \text{child}(k)} \left[P[\widetilde{\beta_{j}^{(i)}}] + N_{h}[\widetilde{\beta_{j}^{(i)}}]\right](0)
\end{equation}
where each polynomial in \Cref{eq:step2-2} is reconstructed by the node $k$.

Finally, to check \Cref{eq:c}, we introduce random polynomials $M_{x,y}$ for $x\in[t^2],y\in[t]$ such that
$$ M_{x,y}(0) = N_x\left[ \alpha_0^{(i)} \right](0) + N_x\left[\beta_0^{(i)} \right](0) - N_y\left[\widetilde{\beta_0^{(i)}}\right](0).$$
Once we have this relation \Cref{eq:c} is checked by
\begin{align*}
        \left[\pal{P}{0}{(i)} +\pal{N_{x}}{0}{(i)}\right](0) +  \left[\pbet{P}{0}{(i)}+\pbet{N_{x}}{0}{(i)}\right](0) -  \left[P\left[\widetilde{\beta_0^{(i)}}\right] + N_{y}\left[\widetilde{\beta_0^{(i)}}\right]\right](0) = M_{x,y}(0).
\end{align*}

\subsection{Zero-knowledge implementation of Step 3 and final application.}

Assuming that the nodes have oracle access to query the value $f(r_1,\dots,r_N)$, the node $k=0$ has the polynomial $\left(\pal{P}{k}{(N)}+ \pal{N_{h^{(N)}}}{k}{(N)}\right)(c_k)$ selected in the final phase of Step 2.

\subsection{Final protocol}\label{subsec:final-protocol}

Now we describe the interaction between the prover and the nodes for each step of our zero-knowledge implementation of Sumcheck and the corresponding verification procedures.

\begin{tcolorbox}[title=Prover-Verifier Interaction $\Pi_1$ for Step 1, enhanced, breakable=true]
\small
\begin{enumerate}[label={\it \arabic*.}]

    \item \textbf{Precomputation} of the honest \textrm{prover} for each node  $k$:
\begin{enumerate}[label={\bf P\arabic*.}]
    \item A random 2-coloring $\{c_k \}_{k\in V}$ of the spanning tree. Then $c_k$ is an uniform random variable in $\{1,2\}$ for all $k\in V$.
    \item Polynomial encryption \( \pal{P}{k}{(1)}, \pbet{P}{k}{(1)}\) and $\left\{\pal{N_h}{k}{(1)}, \pbet{N_h}{k}{(1)}\right\}_{h\in [t^2]}\subseteq_R\field_1[x]$ such that $\pal{N_h}{k}{(1)}(0)~=~-\sum_{j\in \text{child}(k)} \pbet{N_h}{j}{(1)}(0)$ for all $h\in [t^2]$. The prover also computes $\left\{M_h\right\}_{h\in [t^2]}\in_R\field_1[x]$ such that $M_h(0) = \pbet{N_h}{0}{(1)}(0) + \pal{N_h}{0}{(1)}(0)$ for all $h\in [t^2]$.

\end{enumerate}

\item \textbf{Message} from the honest \textrm{prover} to each node $k$:

\begin{enumerate}[label={\bf M\arabic*.}]
    \item Sends the shares $ \pal{P}{k}{(1)}(c_k)$, $\pal{N_h}{k}{(1)}(c_k)$, $\pbet{P}{k}{(1)}(c_k)$, $\pbet{N_h}{k}{(1)}(c_k)$ for each $h\in [t^2]$.
    \item If $p$ is its parent in the ST, the prover also sends the shares $ \pal{P}{p}{(1)}(c_k)$, $ \pal{N_h}{p}{(1)}(c_k)$ and $ \pbet{P}{p}{(1)}(c_k)$, $\pbet{N_h}{p}{(1)}(c_k)$ for each $h\in [t^2]$.
    \item If $k$ has a unique child $j$,  the prover also sends the shares $\pal{P}{j}{(1)}(c_k)$, $\pal{N_h}{j}{(1)}(c_k)$ and $\pbet{P}{j}{(1)}(c_k)$, $\pbet{N_h}{j}{(1)}(c_k)$ for each $h\in [t^2]$.
    \item If $k$ has a next sibling $k'$ in the ST, the prover also sends the shares $\pal{P}{k'}{(1)}(3-c_k)$, $\pal{P_h}{k'}{(1)}(3-c_k)$ and $\pbet{P}{k'}{(1)}(3-c_k)$, $\pbet{N_h}{k'}{(1)}(3-c_k)$ for each $h\in [t^2]$.
    \item If $k=0$, for all $h\in [t^2]$, the prover sends the description of $M_h$ to the node $0$.
\end{enumerate}

\item \textbf{Message} from the root of ST to the \textrm{prover}: Selects $t+1$ indices $h^{(1)}_0,\ldots,h^{(1)}_t\in [t^2]$ uniform at random and sends to the prover. \\

\item \textbf{Message} from the honest \textrm{prover} to each node $k$:  The prover tells $h^{(1)}_0,\ldots,h^{(1)}_t$ to each node.\\

\end{enumerate}

\end{tcolorbox}

\begin{tcolorbox}[title=Verification Protocol of $\Pi_1$,enhanced, breakable=true]
    \small

    \textbf{(Communication)} Each node $k$ in one round:

\begin{enumerate}[label={\bf C\arabic*.}]
    \item Sends to its parent $p$ the shares $\left[ \pbet{P}{k}{(1)} +  \pbet{N_{h_0^{(1)}}}{k}{(1)}\right](c_k)$, $\left[ \pal{P}{p}{(1)} +  \pal{N_{h_0^{(1)}}}{p}{(1)}\right](c_k)$, $\left[ \pbet{P}{p}{(1)} +  \pbet{N_{h_0^{(1)}}}{p}{(1)}\right](c_k)$. 
    For each index $h \in [t^2]\backslash \{h^{(1)}_0,\ldots,h^{(1)}_t \}$, sends to $p$ the shares $\pal{N_h}{k}{(1)}(c_k)$, $\pbet{N_h}{k}{(1)}(c_k)$, $\pal{N_h}{p}{(1)}(c_k)$, and $\pbet{N_h}{p}{(1)}(c_k)$ . If $k$ has the next sibling $k'$, then sends to its parent $p$ the share $\left[ \pbet{P}{k'}{(1)} +  \pbet{N_{h_0^{(1)}}}{k'}{(1)}\right](3-c_k)$, and $\pal{N_h}{k'}{(1)}(3-c_k)$, $\pbet{N_h}{k'}{(1)}(3-c_k)$ for $h \in [t^2]\backslash \{h^{(1)}_0,\ldots,h^{(1)}_t \}$. 
    
\end{enumerate}
\textbf{(Verification)} Each node $k\in V$ verifies the following. 
    \begin{enumerate}[label={\bf V\arabic*.}]
        \item For all $h\in [t^2]\backslash\{h^{(1)}_0,\ldots,h^{(1)}_t\}$, interpolates $\pal{N_h}{k}{(1)}$, $\pbet{N_h}{k}{(1)}$, and $\pbet{N_h}{j}{(1)}$ for $j\in\text{child}(k)$ to check 
        $$\pal{N_h}{k}{(1)}(0) + \sum_{j\in \text{child}(k)}\pbet{N_h}{j}{(1)}(0) = \pbet{N_h}{k}{(1)}(0).$$
        If $k=0$ it also checks for each $h\in [t^2]\backslash\{h^{(1)}_0,\ldots,h^{(1)}_t\}$
        $$M_h(0) = \pal{N_h}{0}{(1)}(0) + \pbet{N_h}{0}{(1)}(0).$$
        \item Interpolates the polynomials $\left[\pbet{P}{k}{(1)} + \pbet{N_{h_0^{(1)}}}{k}{(1)}\right](x)$, $\left[\pal{P}{j}{(1)} + \pal{N_{h_0^{(1)}}}{j}{(1)}\right](x)$ and $\left[\pbet{P}{j}{(1)} + \pbet{N_{h_0^{(1)}}}{j}{(1)}\right](x)$ for each $j\in \text{child}(k)$ and checks that 
        \begin{align*}
                \left[\pal{P}{k}{(1)} + \pal{N_{h_0^{(1)}}}{k}{(1)}\right](0)+ \sum_{j\in \text{child}(k)}\left[\pbet{P}{j}{(1)} + \pbet{N_{h_0^{(1)}}}{j}{(1)}\right](0) \\= \left[\pbet{P}{k}{(1)} + \pbet{N_{h_0^{(1)}}}{k}{(1)}\right](0)
        \end{align*} 

        \item If $k=0$, also verifies the following equality:
 
        \[
        \left[\pal{P}{0}{(1)} +\pal{N_{h_0^{(1)}}}{0}{(1)}\right](0) +  \left[\pbet{P}{0}{(1)}+\pbet{N_{h_0^{(1)}}}{0}{(1)}\right](0)  = M_{h_0^{(1)}}(0)
        \]
    \end{enumerate}
    
\end{tcolorbox}

\begin{tcolorbox}[title=Prover-Verifier Interaction $\Pi_2$ for Step 2, enhanced, breakable=true]
\small
For each $i\in \{2,\ldots, N\}$, the prover and the verifier act as follows.
Rewrite the indices $h^{(i-1)}_1,\ldots, h^{(i-1)}_t \in [t^2]$ from the previous round by $\{1,\ldots,t\}$.

\begin{enumerate}[label={\it \arabic*.}]

\item \textbf{Message} from the root of ST to the \textrm{prover}: Selects $r_{i-1} \in \fieldq$ uniform at random and sends to the prover. \\

    \item \textbf{Local computation} of the honest \textrm{prover} for each node $k$:
\begin{enumerate}[label={\bf P\arabic*.}]
    \item Polynomial encryption \( \pal{P}{k}{(i)}, \pbet{P}{k}{(i)}, P\left[\widetilde{\beta_k^{(i)}}\right]\).
    \item Random polynomials $\left\{\pal{N_h}{k}{(i)}, \pbet{N_h}{k}{(i)}\right\}_{h\in [t^2]}\subseteq_R\field_1[x]$ such that $$\pal{N_h}{k}{(i)}(0)~=~-\Sigma_{j\in \text{child}(k)} \pbet{N_h}{j}{(i)}(0).$$
    \item Random polynomials $\left\{N_h\left[\widetilde{\beta_k^{(i)}}\right]\right\}_{h\in [t]}\subseteq_R\field_1[x]$ such that $$N_{h}\left[\alpha_k^{(i-1)}\right](0) \cdot r_{i-1}^{k}+\sum_{j\in \text{child}(k)} N_h\left[\widetilde{\beta_{j}^{(i)}}\right](0) - N_h\left[\widetilde{\beta_k^{(i)}}\right](0) = 0$$ where we relabeled the indices $\{h_1^{(i-1)},\ldots, h_t^{(i-1)}\}$ by $[t]$ for $N_{h}\left[\alpha_k^{(i-1)}\right](0)$.
    \item Random polynomials $\left\{M_{x,y}\right\}_{x\in [t^2],y\in[t]}\subseteq_R\field_1[x]$ such that $M_{x,y}(0) = \pal{N_x}{0}{(i)}(0) + \pbet{N_x}{0}{(i)}(0) - N_y\left[ \widetilde{\beta_0^{(i)}}\right](0)$.
\end{enumerate}

\item \textbf{Message} from the honest \textrm{prover} to each node $k$:

\begin{enumerate}[label={\bf M\arabic*.}]
    \item Sends the value $r_{i-1}$ and the shares $ \pal{P}{k}{(i)}(c_k)$,  $\pbet{P}{k}{(i)}(c_k)$, $P\left[\widetilde{\beta_{k}^{(i)}}\right](c_k)$, $\pal{N_h}{k}{(i)}(c_k)$, $\pbet{N_h}{k}{(i)}(c_k)$ for each $h \in [t^2]$ and $N_h\left[\widetilde{\beta_{k}^{(i)}}\right](c_k)$ for each $h\in [t]$.
    \item If $p$ is its parent in the ST, the prover also sends the shares $ \pal{P}{p}{(i)}(c_k)$,  $\pbet{P}{p}{(i)}(c_k)$, $P\left[\widetilde{\beta_{p}^{(i)}}\right](c_k)$, $\pal{N_h}{p}{(i)}(c_k)$, $\pbet{N_h}{p}{(i)}(c_k)$ for each $h \in [t^2]$ and $N_h\left[\widetilde{\beta_{p}^{(i)}}\right](c_k)$ for each $h\in [t]$.
    \item If $k$ has a unique child $j$,  the prover also sends the shares $ \pal{P}{j}{(i)}(c_k)$,  $\pbet{P}{j}{(i)}(c_k)$, $P\left[\widetilde{\beta_{j}^{(i)}}\right](c_k)$, $\pal{N_h}{j}{(i)}(c_k)$, $\pbet{N_h}{j}{(i)}(c_k)$ for each $h \in [t^2]$ and $N_h\left[\widetilde{\beta_{j}^{(i)}}\right](c_k)$ for each $h\in [t]$.
    \item If $k$ has a next sibling $k'$ in the ST, the prover also sends the shares $ \pal{P}{k'}{(i)}(c_k)$,  $\pbet{P}{k'}{(i)}(c_k)$, $P\left[\widetilde{\beta_{k'}^{(i)}}\right](c_k)$, $\pal{N_h}{k'}{(i)}(c_k)$, $\pbet{N_h}{k'}{(i)}(c_k)$ for each $h \in [t^2]$ and $N_h\left[\widetilde{\beta_{k'}^{(i)}}\right](c_k)$ for each $h\in [t]$.
    \item If $k=0$, for all $h\in [t]$, the prover sends the description of $M_{x,y}$ for $x\in [t^2]$ and $y\in[t]$.
\end{enumerate}

\item \textbf{Message} from the root of ST to the \textrm{prover}: Selects $t+1$ indices $h^{(i)}_0,\ldots, h^{(i)}_t \in [t^2]$, and an index $h^{(i)}\in [t]$ uniform at random and sends to the prover. \\

\item \textbf{Message} from the honest \textrm{prover} to each node $k$:  The prover tells $h^{(i)}_0,\ldots, h^{(i)}_t \in [t^2]$ and $h^{(i)}$ to each node.\\

\end{enumerate}

\end{tcolorbox}

\begin{tcolorbox}[title=Verification Protocol of $\Pi_2$,enhanced, breakable=true]
    \small
    For each $i\in \{2,\ldots, N\}$, the verification protocol is as follows.

    \textbf{(Communication)} Each node $k$ in one round:

\begin{enumerate}[label={\bf C\arabic*.}]
    \item For $h = h^{(i)}_0$, sends to its parent $p$ the shares $\left[ \pal{P}{k}{(i)} +  \pal{N_{h}}{k}{(i)}\right](c_k)$, $\left[ \pbet{P}{k}{(i)} +  \pbet{N_{h}}{k}{(i)}\right](c_k)$, $\left[ \pal{P}{p}{(i)} +  \pal{N_{h}}{p}{(i)}\right](c_k)$, $\left[ \pbet{P}{p}{(i)} +  \pbet{N_{h}}{p}{(i)}\right](c_k)$ and $\left[ \pbet{P}{k'}{(i)} +  \pbet{N_{h}}{k'}{(i)}\right](3-c_k)$ if $k$ has a next sibling $k'$. For each index $h \in [t^2]\backslash \{h^{(i)}_0,\ldots, h^{(i)}_t\}$, sends to $p$ the shares $\pal{N_h}{k}{(i)}(c_k)$, $\pbet{N_h}{k}{(i)}(c_k)$.
\end{enumerate}
\textbf{(Verification)} Each node $k\in V$ verifies the following. 

    \begin{enumerate}[label={\bf V\arabic*.}]
        \item Checks that for all $h\in [t^2]\backslash \{h^{(i)}_0,\ldots,h^{(i)}_t\}$, $$\pal{N_h}{k}{(i)}(0) + \sum_{j\in \text{child}(k)}\pbet{N_h}{j}{(i)}(0) = \pbet{N_h}{k}{(i)}(0).$$
        \item Checks that for all $h\in [t]\backslash \{h^{(i)}\}$, 
        $$N_h\left[\widetilde{\beta_k^{(i)}}\right](0) = \pal{N_h}{k}{(i-1)}(0)\cdot r_{i-1}^k +  \sum_{j\in \text{child}(k)} N_h\left[\widetilde{\beta^{(i)}_j} \right](0).$$
        \item If $k=0$, checks that for all $x \in [t^2]\backslash \{h^{(i)}_0,\ldots,h^{(i)}_t\}$ and $y\in [t]\backslash \{h^{(i)}\}$,
        $$ M_{x,y}(0) = N_x\left[ \alpha_0^{(i)} \right](0) + N_x\left[\beta_0^{(i)} \right](0) - N_y\left[\widetilde{\beta_0^{(i)}}\right](0).$$
        \item Interpolates the polynomials $\left[\pbet{P}{k}{(i)} + \pbet{N_{h_0^{(i)}}}{k}{(i)}\right](x)$,$\left[\pal{P}{k}{(i)} + \pal{N_{h_0^{(i)}}}{k}{(i)}\right](x)$ and $\left[\pbet{P}{k}{(i)} + \pbet{N_{h_0^{(i)}}}{k}{(i)}\right](x)$ for each $j\in \text{child}(k)$ and checks that 
            \begin{align*}
             \left[\pal{P}{k}{(i)} + \pal{N_{h_0^{(i)}}}{k}{(i)}\right](0)+ \sum_{j\in \text{child}(k)}\left[\pbet{P}{j}{(i)} + \pbet{N_{h_0^{(i)}}}{j}{(i)}\right](0) 
             \\= \left[\pbet{P}{k}{(i)} + \pbet{N_{h_0^{(i)}}}{k}{(i)}\right](0)
            \end{align*}
        \item Interpolates the polynomials $\left[P\left[\alpha_k^{(i-1)}\right] + N_{h^{(i)}}\left[\alpha_k^{(i-1)}\right]\right](x)$, $\left[P\left[\widetilde{\beta_k^{(i)}}\right] + N_{h^{(i)}}\left[\widetilde{\beta_k^{(i)}}\right]\right](x)$ and $\left[P\left[\widetilde{\beta_j^{(i)}}\right] + N_{h^{(i)}}\left[\widetilde{\beta_j^{(i)}}\right]\right](x)$ for each $j\in \text{child}(k)$ and checks that
        \begin{align*}
            \left[P\left[\alpha_k^{(i-1)}\right] + N_{h^{(i)}}\left[\alpha_k^{(i-1)}\right]\right](0) \cdot r_{i-1}^k + \sum_{j\in\text{child}(k)}
            \left[P\left[\widetilde{\beta_j^{(i)}}\right] + N_{h^{(i)}}\left[\widetilde{\beta_j^{(i)}}\right]\right](0) 
            \\ = \left[P\left[\widetilde{\beta_k^{(i)}}\right] + N_{h^{(i)}}\left[\widetilde{\beta_k^{(i)}}\right]\right](0) 
        \end{align*}
        \item If $k=0$, also verifies the following equality:
 
        \begin{align*}
        \left[\pal{P}{0}{(i)} +\pal{N_{h^{(i)}_0}}{0}{(i)}\right](0) +  \left[\pbet{P}{0}{(i)}+\pbet{N_{h^{(i)}_0}}{0}{(i)}\right](0) -  \left[P\left[\widetilde{\beta_0^{(i)}}\right] + N_{h^{(i)}}\left[\widetilde{\beta_0^{(i)}}\right]\right](0) \\= M_{h^{(i)}_0,h^{(i)}}(0)
        \end{align*}
    \end{enumerate}
    
\end{tcolorbox}

\begin{tcolorbox}[title=Prover-Verifier Interaction $\Pi_3$ for Final Check, enhanced, breakable=true]
\small

Rewrite the indices $h^{(N)}_1,\ldots, h^{(N)}_t \in [t^2]$ from the previous round by $\{1,\ldots,t\}$.

\begin{enumerate}[label={\it \arabic*.}]

\item \textbf{Message} from the root of ST to the \textrm{prover}: Selects $r_N\in \fieldq$ uniformly at random, and sends to the prover.\\

 \item \textbf{Local computation} of the honest \textrm{prover} for each node $k$:
\begin{enumerate}[label={\bf P\arabic*.}]
    \item Polynomial encryption \( P\left[\widetilde{\beta_k^{(N+1)}}\right]\).
    \item random polynomials $\left\{N_h\left[\widetilde{\beta_k^{(N+1)}}\right]\right\}_{h\in [t]}\subseteq_R\field_1[x]$ such that 
    $$N_{h}\left[\alpha_k^{(N)}\right](0) \cdot r_{N}^{k}+\sum_{j\in \text{child}(k)} N_h\left[\widetilde{\beta_{j}^{(N+1)}}\right](0) - N_h\left[\widetilde{\beta_k^{(N+1)}}\right](0) = 0$$ 
    where $\left\{N_{h}\left[\alpha_k^{(N)}\right](0)\right\}_{h\in [t]}$ were determined in the previous round. Especially for $k=0$, set $N_h\left[\widetilde{\beta_0^{(N)}}\right](0) =0$.
\end{enumerate}

\item \textbf{Message} from the honest \textrm{prover} to each node $k$:

\begin{enumerate}[label={\bf M\arabic*.}]
    \item Sends the shares $ P\left[\widetilde{\beta_k^{(N+1)}}\right](c_k)$, $N_h\left[\widetilde{\beta_k^{(N+1)}}\right](c_k)$ for each $h\in [t]$.
    \item If $p$ is its parent in the ST, the prover also sends the shares $ P\left[\widetilde{\beta_p^{(N+1)}}\right](c_k)$, $N_h\left[\widetilde{\beta_p^{(N+1)}}\right](c_k)$ for each $h\in [t]$.
    \item If $k$ has a unique child $j$,  the prover also sends the shares $ P\left[\widetilde{\beta_j^{(N+1)}}\right](c_k)$, $N_h\left[\widetilde{\beta_j^{(N+1)}}\right](c_k)$ for each $h\in [t]$.
    \item If $k$ has a next sibling $k'$ in the ST, the prover also sends the shares $ P\left[\widetilde{\beta_{k'}^{(N+1)}}\right](3-c_k)$, $N_h\left[\widetilde{\beta_{k'}^{(N+1)}}\right](3-c_k)$ for each $h\in [t]$.
\end{enumerate}

\item \textbf{Message} from the root of ST to the \textrm{prover}: Selects an index $h^{(N+1)} \in [t]$ uniform at random and sends to the prover. \\

\item \textbf{Message} from the honest \textrm{prover} to each node $k$:  The prover tells $h^{(N+1)}$ to each node.\\

\end{enumerate}

\end{tcolorbox}

\begin{tcolorbox}[title=Verification Protocol of $\Pi_3$,enhanced, breakable=true]
    \small
    For each $i\in \{2,\ldots, N\}$, the verification protocol is as follows.

    \textbf{(Communication)} Each node $k$ in one round:

\begin{enumerate}[label={\bf C\arabic*.}]
    \item Sends to its parent $p$ the shares $\left[ \pal{P}{k}{(N)} +  \pal{N_{h^{(N+1)}}}{k}{(N)}\right](c_k)$, $\left[ \pal{P}{p}{(N)} +  \pal{N_{h}}{p}{(N)}\right](c_k)$ and $\left[ \pal{P}{k'}{(N)} +  \pbet{N_{h}}{k'}{(N)}\right](3-c_k)$ if $k$ has a next sibling $k'$. For each index $h \in [t]\backslash \{h^{(N+1)}\}$, sends to $p$ the shares $\pal{N_h}{k}{(i)}(c_k)$, $\pbet{N_h}{k}{(i)}(c_k)$.
    
\end{enumerate}
\textbf{(Verification)} Each node $k\in V$ verifies the following. 

    \begin{enumerate}[label={\bf V\arabic*.}]
        \item Checks that for all $h\in [t]\backslash \{h^{(N+1)}\}$, 
        $$N_h\left[\widetilde{\beta_k^{(N+1)}}\right](0) = \pal{N_h}{k}{(N)}(0)\cdot r_{N}^k +  \sum_{j\in \text{child}(k)} N_h\left[\widetilde{\beta^{(N+1)}_j} \right](0),$$
        and if $k=0$, also checks that $N_h\left[\widetilde{\beta_k^{(N+1)}}\right](0) = 0.$
        \item Interpolates the polynomials $\left[P\left[\widetilde{\beta_k^{(N+1)}}\right] + N_{h^{(N+1)}}\left[\widetilde{\beta_k^{(N+1)}}\right]\right](x)$, $\left[P\left[\alpha_k^{(N)}\right] + N_{h^{(N+1)}}\left[\alpha_k^{(N)}\right]\right](x)$, $\left[P\left[\widetilde{\beta_j^{(N+1)}}\right] + N_{h^{(N+1)}}\left[\widetilde{\beta_j^{(N+1)}}\right]\right](x)$ for each $j\in \text{child}(k)$ and checks that 
            \begin{align*}
            \left[P\left[\alpha_k^{(N)}\right] + N_{h^{(N+1)}}\left[\alpha_k^{(N)}\right]\right](0)\cdot r_N^k +
            \sum_{j\in \text{child}(k)} \left[P\left[\widetilde{\beta_j^{(N+1)}}\right] +  N_{h^{(N+1)}}\left[\widetilde{\beta_j^{(N+1)}}\right]\right](0) \\
            = \left[P\left[\widetilde{\beta_k^{(N+1)}}\right] + N_{h^{(N+1)}}\left[\widetilde{\beta_k^{(N+1)}}\right]\right](0)
            \end{align*}
   
        \item If $k=0$, checks the following equality:
 
        \begin{align*}
        \left[P\left[\widetilde{\beta_0^{(N+1)}}\right] + N_{h^{(N+1)}}\left[\widetilde{\beta_k^{(N+1)}}\right]\right](0) = f(r_1,\ldots,r_N)
        \end{align*}
        by using oracle access to $f$. 
    \end{enumerate}
    
\end{tcolorbox}

\subsection{Completeness and Soundness}
Let us now analyze the protocol presented in~\Cref{subsec:final-protocol}. Assume that the prover follows the honest strategy. 
\begin{description}
    \item[Protocol $\Pi_1$:] In {\bf V2.} of $\Pi_1$, each node $k$ checks $\alpha_k^{(1)}+\sum_{j\in\text{child}(k)}\beta_j^{(1)} = \beta_k^{(1)}$, and in {\bf V3.} of $\Pi_1$, the node $k=0$ also checks $\alpha_0^{(1)}+\beta_0^{(1)}=0$. 
    \item[Protocol $\Pi_2$:] In {\bf V4.} and {\bf V5.} of $\Pi_2$, each node $k$ checks $\alpha_k^{(i)}+\sum_{j\in\text{child}(k)}\beta_j^{(i)} = \beta_k^{(i)}$ and $\alpha_k^{(i-1)}\cdot r_{i-1}^k + \sum_{j\in \text{child}(k)}\widetilde{\beta_j^{(i)}} = \widetilde{\beta_k^{(i)}}$, respectively. In {\bf V6.} of $\Pi_2$, the node $k=0$ also checks $\alpha_0^{(i)}+\beta_0^{(i)}=\widetilde{\beta_0^{(i)}}$. 
    \item[Protocol $\Pi_3$:] In {\bf V2.} of $\Pi_3$, each node $k$ checks $\widetilde{\beta_k^{(N+1)}} = \alpha_k^{(N)}\cdot r_n^k + \sum_{j\in\text{child}(k)}\widetilde{\beta_j^{(N+1)}}$, and in {\bf V3.} of $\Pi_3$, the node $k=0$ also checks $\widetilde{\beta_0^{(N+1)}}=f(r_1,\ldots,r_N)$. 
\end{description}
To summarize, $\Pi_1$, $\Pi_2$, $\Pi_3$ exactly simulate the distributed Sumcheck protocol.
The completeness thus follows directly from the completeness of the Sumcheck protocol. 

For soundness, we first observe that if the malicious prover sends the random polynomials that do not satisfy the conditions checked in the verification protocol (more precisely, {\bf V1.} in the verification protocol of $\Pi_1$, {\bf V1., V2., V3} in the verification protocol of $\Pi_2$, and {\bf V1.} in the verification protocol of $\Pi_3$), the acceptance probability is at most $1/t$. We thus focus on the case that the prover sends the correct random polynomials. In this case, by construction of our protocol the nodes check the same conditions checked in the distributed Sumcheck described in Section~\ref{sec:distsc}, and thus the acceptance probability is at most $\frac{N\cdot\mathrm{deg}(F)}{q}$ where $\mathrm{deg}(F)$ is individual degree of $F$, by the soundness of Sumcheck~(Theorem~\ref{thm:completeness-soundness-sumcheck}). 

\paragraph*{The existence of the simulator.} With these detailed analysis, it is straightforward to observe that all the messages received from the prover and the neighbors are masked by random polynomials, hence uniformly at random. Therefore, the simulator merely generates uniform random values for its output (and a single oracle query to $F$), which exactly simulates the distribution of the real transcript. It is obvious that the verification algorithm is in the class $\mathcal{C}[\bigo{1}, \bigo{N \log q}]$ since all the verification parts in the protocol $\Pi_1$, $\Pi_2$, and $\Pi_3$ can be parallelized. Moreover, it is easy to see the simulator is also in the same class: it require only $\bigo{1}$ rounds of communication between neighboring nodes to allow each node to learn the IDs of its parent and children, as well as its next sibling (if any) in the spanning tree. The rest of the simulator simply samples uniformly at random values (satisfying required conditions), and a single oracle query to $F$ used for Final Check.

This completes the proof of~\Cref{theo:zksc}.

\section{Applications of Distributed Sumcheck}\label{sec:appsc}

In this section, we prove \Cref{theo:informalzkapp1} and \Cref{theo:informalzkapp2}, and provide formal statements of these results in terms of the statistical zero-knowledge dIP (dStatZK) class defined in \Cref{def:statzk}. 

To address these applications, we do not treat oracle access to the polynomials defined in \nonthreecol\ and \subgraphcounting\ as a black box; instead, we describe how these evaluations can be performed. Recall that we assume that each node knows the identifier of its parent in an arbitrary rooted spanning tree written in the input label, as mentioned in the definition of distributed languages in~\Cref{subsec:dipandzk}. Therefore, we can explicitly evaluate each polynomial by the nodes along the spanning tree with the assistance of the prover (see \Cref{sec:appsc} for details). Consequently, to establish \Cref{theo:informalzkapp1} and \Cref{theo:informalzkapp2}, it remains only to show how these evaluations can be simulated by the simulator. The key result of this section, and the reason why our protocol achieves statistical zero-knowledge rather than perfect zero-knowledge, is that we can guarantee only that each message used to compute one evaluation of the polynomials is statistically indistinguishable from the uniform distribution.

\subsection{\nonkcol}\label{subsec:dip-for-non-3-col}
\label{subsec:non3col}
In this section, we study the problem \nonkcol. The formal definition of the problem is below. 
\begin{algorithmbox}{\nonkcol}{}{
    \textbf{Input.} An $n$-node graph $G=(V,E)$, with unique identifiers $\{0,\ldots, n-1\}$. Each node knows fixed integer $k$ and the ID of its parent in an arbitrary spanning tree rooted at node with ID $0$.\\
    \textbf{Decision problem:} Decide if the graph $G$ does not admit a $k$-coloring.
}
\end{algorithmbox}

Here given a graph $G=(V,E)$, a $k$-coloring of $G$ is a function $c\colon V\to \{1,\dots,k\}$ such that for all $\{u,v\}\in E$, $c(u)\neq c(v)$. 

We prove the following theorem.

\begin{theorem}\label{theo:finalnon3zk}
   Let $q\in n^{\omega(1)}$ be any prime. The language $\nonthreecol$ is in $$\mathrm{dStatZK}\left[\bigo{n},\bigo{\log q}, C[\bigo{1},\bigo{n\log q} ]\right].$$
\end{theorem}

Before proving \Cref{theo:finalnon3zk}, we prove that \nonkcol\ problem can be reduced to a Sumcheck instance.

\begin{lemma}\label{lem:rednonkcol}
    Given an instance $(G,I)$, where $n=|V(G)|$ and $m=|E(G)|$, there exists a polynomial $P_G\colon \{0,1\}^{N}\to\{0,1\}$ such that
    \[G \text{ is not }k \text{ colorable}\iff \sum_{x\in\{0,1\}^N}P_G(x) = 0,\]
    where $N$ is $\bigo{n+m}$.
\end{lemma}

\begin{proof}
Consider the classical reduction of \nonkcol\ to $k$-SAT: Given a graph  $G~=~(V,E)$ with $n$ vertices and $m$ edges, consider the following $k$-CNF formula $\phi(G)$ defined as follows.
\begin{itemize}
    \item Variables of $\phi(G)$: For each node $v\in V$ the formula $\phi$ has $k$ variables $c_v^1,\dots, c_v^k$ (one for each color).
    \item Clauses of $\phi(G)$: The formula $\phi(G)$ consists of two types of clauses:

\begin{enumerate}
\item Color Selection Constraints for each vertex $v$:

\begin{equation*}
\underbrace{\bigvee_{i\in [k]}c_v^{i}}_{\text{at least one color}} \land
\underbrace{\bigwedge_{i\neq j, i,j\in [k]}(\neg c_v^{i}\vee \neg c_v^{j})}_{\text{(at most one color)}}
\end{equation*}

\item Adjacency Constraints for each edge $(u,v) \in E$:
\begin{equation*}
\bigwedge_{i,j\in[k], i\neq j}(\neg c_v^{i} \lor \neg c_v^j)
\end{equation*}
\end{enumerate}
\end{itemize}
By construction of $\phi(G)$ from the graph $G$, the following known fact holds.

\begin{fact}
$G$ is $k$-colorable iff $\phi(G)$ is satisfiable.
\end{fact}

The formula $\phi$ has $N=k\cdot n$ variables and $M=(k+1)\cdot n + k\cdot m$ clauses. Given such a $k$-CNF formula $\phi$ with $M$ clauses over $N$ variables, we construct a multivariate polynomial $P_{G}: \{0,1\}^N \to \{0,1\}$ such that:
\[
P_G(x_{1},\dots, x_N) = 1 \iff \phi \text{ is satisfied by assignment } x_{1},\dots, x_{N}
\]
where $x\in\{0,1\}^N$ is a possible evaluation of the $k\cdot n$ variables $\{c_v^{i}\}_{v\in V, i\in [k]}$. The polynomial $P_G$ is constructed as the usual arithmetization of the formula $\phi$.

\begin{itemize}
\item
For each edge \(e=(u,v)\in E\) with $\id(u)<\id(v)$ and color $i\in [k]$, define the polynomial
\[
   W^{i}_e(c_u^{i}, c_v^{i}):=1-c_u^{i}c_v^{i}.
\]

and then the polynomial $T(x)\colon \{0,1\}^N\to \{0,1\}$ given by

$$T(x):=\prod_{i\in [k]}\prod_{e=uv\in E}W^{i}_e(c_u^{i}, c_v^{i})$$

Where $x\in \{0,1\}^N$ contains the $N$ variables of the form $c_v^{i}$ for all $v\in V$ and $i\in [k]$.
\item
For each vertex $v\in V$, define the polynomials $A_v, B_v\colon \{0,1\}^k\to \{0,1\}$ as
\[
   A_v(c_v^1,\dots, c_v^k):=1- \prod_{i\in [k]}(1-c_v^{i}),\quad\text{and \quad}
   B_v:(c_v^{1},\dots, c_v^k)=\prod_{i,j\in [k], i<j}(1-c_v^{i} c_v^j).
\]
and the polynomial $S\colon \{0,1\}^N\to \{0,1\}$ as
\[
   S(x):=\prod_{v\in V}A_v(c_v^1,\dots, c_v^k)\cdot B_v(c_v^1,\dots, c_v^k).
\]
\end{itemize}
Finally, defining the polynomial $P_G(x)\colon\{0,1\}^N\to\{0,1\}$ as
$
   P_{G}(x):=T(x)\cdot S(x)
$
 yields the following claim, which concludes the proof. 

\begin{claim}\label{lemm:ssfor3col}
    The graph $G$ is not $k$-colorable iff 

    \[\sum_{x\in \{0,1\}^{N}}P_{G}(x) = 0\]
\end{claim}

\begin{proof}[Proof of \Cref{lemm:ssfor3col}]
    Direct from \Cref{fact:interpol}, since each evaluation $P_G(x)$ for a vector $x\in\{0,1\}^N$ correspond to the evaluation $\phi(x)$ and this formula is satisfiable iff $x$ is a $k$-coloring of the graph.
\end{proof}

\end{proof}

Given this derivation, we proceed to prove \Cref{theo:finalnon3zk}.

\begin{proof}[Proof of \Cref{theo:finalnon3zk}.] Let $G=(V,E)$ by an $n$-node graph, and consider the polynomial \( P_{G} \) given by \Cref{lem:rednonkcol}. We can employ the Sumcheck protocol to verify the absence of a satisfying assignment, i.e., whether:
\[
\sum_{\mathbf{x} \in \{0,1\}^N} P_G(\mathbf{x}) = 0.
\]

In order to use \Cref{theo:zksc}, we need that the addition $\sum_{\mathbf{x} \in \{0,1\}^N} P_G(\mathbf{x})$ is at most a value $k=\bigo{poly(n)}$\footnote{In no-instances, the sum can be very large} since otherwise cannot be encoded in $\bigo{\log n}$, but at this point the sum of $P_G$ can be as large as $2^{\Omega(n)}$. 
 Therefore, in the Sumcheck protocol for the polynomial $P_G$, we have to assume a field $\mathbb{F}_q$ of size $q= \Omega(2^n)$ to ensure
 \begin{align}\label{eq}
     \sum_{\mathbf{x}\in\{0,1\}^N}P_G(\mathbf{x})=0 \Leftrightarrow \sum_{\mathbf{x}\in\{0,1\}^N}P_G(\mathbf{x})=0 \text{ mod $q$}.
 \end{align}
This would violate our $\mathcal{O}(\log n)$ message size constraint. 
To address this issue, we proceed as follows: 
instead of selecting a prime $q$ of exponential size in $n$, 
we choose a prime $q$ uniformly at random from the interval $[R,2R]$, 
where $R = \mathrm{poly}(n)$, and apply the Prime Number Theorem.

\begin{theorem}[Prime Number Theorem]
Let $\pi(x)$ denote the prime-counting function, i.e., the number of primes less than or equal to $x$. Then
\[
 \pi(x) \sim \frac{x}{\log x}\quad\text{ as }x\to \infty
\]

\end{theorem}
By the Prime Number Theorem, the interval $[R,2R]$ contains 
$\Theta\!\left(\tfrac{R}{\log R}\right)$ primes. 
Furthermore, since 
\[
\sum_{\mathbf{x}\in\{0,1\}^N} P_G(\mathbf{x}) \leq 2^N,
\]
it can have at most $\bigo{N}$ distinct prime factors. 
Combining these two facts, we obtain

\[
\mathbb{P}_{q}\!\left[\, q \;\text{divides } \sum_{\mathbf{x}\in\{0,1\}^N} P_G(\mathbf{x}) \,\right] 
   \leq \Theta\!\left(\frac{N}{R/\log R}\right) 
   \leq \frac{1}{\mathrm{poly}(n)}.
\]
Therefore, condition~\eqref{eq} holds with high probability over the random choice of $q$, 
and representing each field element in $\mathbb{F}_q$ requires only 
$\bigo{\log n}$ bits (rather than $\mathrm{poly}(n)$ bits). Then we conclude by \Cref{theo:zksc}, since verifying whether a graph is in \nonkcol is equivalent to checking whether Equation~\eqref{eq} holds in the Sumcheck-instance $(G,P_G,q,0)$,  and replacing the oracle access by the following distributed oracle $P_G$.

\paragraph*{Distributed oracle access to $P_G$.}\label{para:non3col}
The oracle access is defined as follows. 
Given inputs $\textbf{x} = x_1,\ldots,x_N$, each node can locally compute the polynomial $S(\textbf{x})$ without interaction with its neighbors. Then, it only remains to specify how the value $T(\textbf{x})$ will be computed.

Given $\textbf{x}$, each node $u$ can compute locally the following part of $T(\textbf{x})$ :
\[
T_u(\textbf{x}) := \prod_{\substack{e=(u,v)\in E \\ v \in T_u,\; \id(u)<\id(v)}}\prod_{i\in [k]} 
   W^{i}_e(\textbf{x})
\]
and it can forward it to its parent in the spanning tree. 
Then, for each child $v$ sending the value $T_v(\textbf{x})$, node $u$ verifies the following consistency condition:
\[
T_u(\textbf{x}) = \prod_{v \in \mathrm{Child}(u)} T_v(\textbf{x}) \cdot
       \prod_{\substack{e=(u,v)\in E \\ v \in N(u),\; \id(u)<\id(v)}}\prod_{i\in[k]}
       W^{i}_e 
\]

 Therefore, at the end of this computation, if $u$ is the root node of the spanning tree, $T_u$ is equivalent to $T(x_1,\dots, x_N)$. Since $S(x_1,\dots, x_N)$ can be computed locally by $u$, the root node can compute $P_G(x_1,\dots, x_N)$.

{\it Statistical simulator.} Each node $v\in V$ throughout the protocol receives two kind of messages:

\begin{itemize}
    \item The messages given by the prover in Steps 1. and 2. of Sumcheck protocol, which by definition are all of them uniform random variables of some known (for node $v$) set.
    \item The messages $T_u(\textbf{x})$ given by each child $u$ of the node $v$, as specified in the distributed oracle access to polynomial $P_G$.
\end{itemize}

And all the information learned by node $v$ is either these messages themselves or some function of these messages. 

To conclude, we prove that the evaluation of the polynomial $P_G(x):=T(x)\cdot S(x)$ at a randomly chosen point (as described by the above distributed oracle) is statistically indistinguishable from the uniform distribution and therefore the oracle query can be replaced with an uniform random value. To this end, observe that the vertex factor $S(x)$ can be determined without any information, and therefore each node can locally simulate uniform random variables and evaluate them in its $S(x)$. The following lemma shows that the evaluation of each $T_v(x)$ can be replaced with an uniform random value in the simulator, as in fact we prove that the complete polynomial $T(\textbf{x})$ is statistically indistinguishable from the uniform distribution.

\begin{lemma}\label{thm:bias-P}
Let $G=(V,E)$ be a graph with $n=|V|,\;m=|E|$ and an integer $k$.
Draw independently
$c^1_v,\dots,c_v^k\sim\Unif(\field_q)$ for each node $v\in V$.
Consider the product of edge-factors $T$ appear in the polynomial $P_G$ above.
Then
\[
   \TVc{T}{\Unif(\field_q)}
     \le O\left(\frac{n^2}{q}\right).
\]
\end{lemma}

\begin{proof}[Proof of \Cref{thm:bias-P}]

Lets define $P^{(0)}:=T$ and for each $i=1,...,L=k\cdot m$, we replace the $i$-th factor of $P^{(i-1)}$ by a new variable $U_i\sim\Unif(\mathbb F_q)$ and call it $P^{(i)}$. Notice that the factors of $T$ are of the form $1-xy$ for two random variables $x,y\overset{\textit{i.i.d.}}{\sim}\Unif(\fieldq)$ and $P^{(L)} = \displaystyle\prod_{i=1}^L U_i$ with $U_i\overset{\textit{i.i.d.}}{\sim}\Unif(\fieldq)$ for all $i\in [L]$.

For a given $P^{(i)}$, let $G:=1-xy$ be the factor of $P^{(i)}$ which will be replaced by $U_i\sim \Unif(\mathbb F_q)$ in $P^{(i+1)}$. Let $R_G$ be the product of remaining factors of $P^{(i)}$ containing $x$ or $y$, and $R_{ind}$ be the product of factors that are independent of $G$. Then, $P^{(i)}=G\cdot R_G \cdot R_{ind}$ and $P^{(i+1)}=U_i\cdot R_G \cdot R_{ind}$.

By \Cref{lem:product-uniforms}, since $R_{ind}$ is independent of $G\cdot R_G$ and $U_i\cdot R_G$,
\begin{align}\label{inequality:1}
    \TV{P^{(i)}}{P^{(i+1)}}
    &= \TV{G\cdot R_G\cdot R_{ind}}{U_i\cdot R_G\cdot R_{ind}}\nonumber \\
    &\leq \TV{G\cdot R_G}{U_i\cdot R_G}.
\end{align}
Therefore we focus on $\TVc{G\cdot R_G}{U_i\cdot R_G}$. First, observe that a direct calculation shows that $\TVc{G}{\Unif(\fieldq)}~=~\frac{q-1}{q^2}\in\bigoo(1/q)$. Thus, if there is no factor corresponding to $R_G$ (the case where all factors in $P^{(i)}$ is independent from $x$ and $y$), we can bound $\TV{P^{(i)}}{P^{(i+1)}}=\bigo (1/q)$. We thus assume that there is a factor corresponding to $R_G$.
Without loss of generality, we assume that $x=c_u^1$ and and $y=c_v^1$ (the other $k-1$ where $(x,y) = (c_u^j,c_v^j)$ for some $j\in [k]\setminus\{1\}$ are analogous). For any $s,t \in \field$, let $\mathcal{E}_{s,t}$ be the event such that $c_u^1 = s$ and $c_v^1=t$. Then, using \Cref{lem:partition} we have

\begin{align}\label{eq:lastnedd}
    \TV{P^{(i)}}{P^{(i+1)}}) &\leq \TVc{G\cdot R_G}{U_i\cdot R_G} \nonumber \\
    &=\sum_{s,t \in \mathbb F_q} \Prob{\mathcal{E}_{s,t}}\cdot \TVc{(1-st)\cdot R_G\middle|\mathcal E_{s,t}}{U_i\cdot R_G\middle|\mathcal E_{s,t}}\nonumber\\
    &\le \sum_{st \notin \{0,1\}} \Prob{\mathcal{E}_{s,t}}\cdot \TVc{(1-st)\cdot R_G\middle|\mathcal E_{s,t}}{U_i\cdot R_G\middle|\mathcal E_{s,t}} + \Prob{r_ur_v\in \{0,1\}}.
\end{align}
Notice that each factor appearing in $R_G$ is either
\begin{itemize}
    %\item $A_u, A_v,B_v,B_u$.
    %\item $P=Q_w \colon = (1-r_ur_w)(1-r_vr_w)$ for each $w\in N(u)\cap N(v)$.
    \item $L^{u,v}_w \colon = (1-c_u^1c_w^1)(1-c_v^1c_w^1)$ for each $w\in N(u)\cap N(v)$,
    \item $L_w^u \colon = (1-c_u^1c_w^1)$ for each $w\in N(u)\backslash N(v)$, or
    \item $L_w^v \colon = (1-c_v^1c_w^1)$ for each $w\in N(v)\backslash N(u)$,
\end{itemize}
and $R_G$ is the product of these $\bigoo(n)$ factors.  When $c_u^1=s$ and $c_v^1=t$ such that $s\cdot t \notin \{0,1\}$, both $L_w^u$ and $L_w^v$ are exactly uniform random variables, and for each $w\in N(u)\cap N(v)$, the probability that $L^{u,v}_w=0$ is at most $\bigo{1/q}$. By union bound, the probability that $\prod_w L^{u,v}_w\neq 0$ is at least $1-\bigo{n/q}$. 
Moreover, both $L_w^u$ and $L_w^v$ are independently uniform, when $r_u=s$ and $r_v=t$ such that $st \notin \{0,1\}$. Thus, by~\Cref{lem:product-uniforms},
\begin{align*}
    \TVc{R_G\middle|\mathcal E_{s,t}}{\Unif(\fieldq)}&\le \Prob{\prod_w L^{u,v}_w \neq 0} \cdot  \TVc{R_G\middle|\mathcal E_{s,t}}{\Unif(\fieldq)} +  \Prob{\prod_w L^{u,v}_w = 0}\\
    &\le \left(1- \bigo{\frac{n}{q}} \right)\cdot \bigo{\frac{n}{q}}  +  \bigo{\frac{n}{q}} \\ &= \bigo{\frac{n}{q}}
\end{align*}
for every $s,t$ such that $s\cdot t\notin \{0,1\}$.  Again by~\Cref{lem:product-uniforms},
\begin{align*}
     &\TVc{(1-st)\cdot R_G\middle|\mathcal E_{s,t}}{U_i\cdot R_G\middle|\mathcal E_{s,t}}\\
     &\le \TVc{(1-st)\cdot R_G\middle|\mathcal E_{s,t}}{(1-st)\cdot \Unif(\fieldq)} + \TVc{(1-st)\cdot \Unif(\fieldq)}{U_i\cdot \Unif(\fieldq)} \\
     & \hspace{5mm} +  \TVc{U_i\cdot \Unif(\fieldq)}{U_i\cdot R_G\middle|\mathcal E_{s,t}} \\
     & = \bigo{n/q} + \bigo{1/q}  + \bigo{n/q}\\ &= \bigo{n/q}.
\end{align*}
We also have that for fixed $s$ and $t$, $\Prob{\mathcal{E}_{s,t}}=\bigo{1/q^2}$ and $ \Prob{c_u^1\cdot c_v^1\in \{0,1\}}=\bigo{1/q}$. Together with~\Cref{eq:lastnedd}, we get $\TV{P^{(i)}}{P^{(i+1)}}) = O\left(1/q\right)$ assuming $q=n^{\omega(1)}$.

Now using the triangle inequality and \Cref{lem:product-uniforms}, the bias of edge factors $T$ is bounded as follows.
\begin{align}
    \TVc{T}{\Unif(\fieldq)}&\leq\TVc{P^{(0)}}{P^{(L)}}+ \TVc{P^{(L)}}{\Unif(\fieldq)}\\
    &\leq \sum_{i=1}^L\TVc{P^{(i-1)}}{P^{(i)}} +\TVc{P^{(L)}}{\Unif(\fieldq)} \\
    &\leq\bigo{\dfrac{L}{q}} + \bigo{\dfrac{L}{q}} = \bigo{\dfrac{n^2}{q}}.
\end{align}
\end{proof}

  Moreover, since the node $v$ computes some portion $T_v(x)$ of the polynomial $T(x)$ using the messages given by the prover and its neighbors ($T_v(x)$ can be obtained by eliminating some factors from $T(x)$),  the same argument works for the distribution of $T_v(x)$, showing each $T_v(x)$ is also statistically close to uniform. Overall, the constant round simulator just needs to generate enough uniform random variables for each node, since all the information received by each node is either an uniform random variable itself, o it is statistically indistinguishable from uniform distribution (the messages $T_v$) given \Cref{thm:bias-P}. 
  
  Together with~\Cref{theo:zksc},  this concludes the proof of~\Cref{theo:finalnon3zk}.

 \end{proof}
 
\subsection{Counting Subgraphs.}
In this section, we study the problem \subgraphcounting. The formal definition of the problem is below. 
\begin{algorithmbox}{\subgraphcounting}{}{
    \textbf{Input.} An $n$-node graph $G=(V,E)$, with unique identifiers $\{0,\ldots, n-1\}$. Each node knows a $k$-node pattern graph $H$, the number $\Delta <n^k$, and the ID of its parent in an arbitrary spanning tree rooted at node with ID $0$.\\
    \textbf{Decision problem:} Decide if the number of copies of $H$ in $G$ is exactly $\Delta$ or not.
}
\end{algorithmbox}

For this problem, we show the following result, which is a formalization of  \Cref{theo:naivetriangle}.

\begin{theorem}\label{theo:finaltczk} 
    The language  $\subgraphcounting$ is in $$\mathrm{dStatZK}\left[\bigo{k\log n},\bigo{\log q}, \mathcal C[\bigo{1},\bigo{\log q}]\right]$$
    where $q$ is any prime such that $q\in n^{\omega(1)}$ for $k\in\bigoo(1)$ and $q>n^k$ for $k\in\omega(1)$.
\end{theorem}

\begin{proof}[Proof of \Cref{theo:finaltczk}]

Given a graph $G=(V,E)$, let $A:\{0,1\}^{\log n} \times \{0,1\}^{\log n}\rightarrow \{0,1\}$ be the adjacency matrix of $G$ as a function, i.e., $A(i,j) =1 \iff \{i,j\}\in E$ and $A(i,j) = 0$ otherwise.

For a given $H=(V(H),E(H))$, the number of $H$'s $\Delta$ is represented as
\begin{align}\label{eq:subgraph}
\mathsf{Aut}(H)\cdot\Delta = \sum_{v_1,\ldots,v_k\in \{0,1\}^{\log n}} \prod_{(v_i,v_j)\in E(H)} A(v_i,v_j),
\end{align}
where $\mathsf{Aut}(H)$ is the number of automorphisms of $H$.
When dealing with induced copies, the representation also includes the non-edge factors
\begin{align*}
\mathsf{Aut}(H)\cdot\Delta = \sum_{v_1,\ldots,v_k\in \{0,1\}^{\log n}} \prod_{(v_i,v_j)\in E(H)} A(v_i,v_j)\prod_{(v_i,v_j)\notin E(H)} \left(1-A(v_i,v_j)\right).
\end{align*}
We will check the right-hand side of the equation using the Sumcheck protocol. Let $\widetilde{A}: \mathbb F_q^{2\log n}\rightarrow \mathbb F_q$ be the multilinear extension of $A$, i.e.,
\[
\widetilde{A}(x,y) = \sum_{z,w \in \{0,1\}^{\log n} } A(z,w) \cdot 
\chi_{x,y}(z,w),
\]
where $\chi_{x,y}$ is a polynomial defined as follows ($x_{\ell}$ is the $\ell$-th element of the $|\mathbb F_q|$-ary representation of $x\in \mathbb F_q^{\log n}$):
\[
\chi_{x,y}(z,w) = \prod_{\ell \in \{1,\ldots, \log n\}}\left[ x_\ell \cdot z_{\ell} + (1-x_{\ell}) \cdot (1-z_{\ell})  \right] \cdot
\left[ y_\ell \cdot w_{\ell} + (1-y_{\ell}) \cdot (1-w_{\ell})  \right].
\]

For simplicity, we only focus on the case of $k$-cliques. For all (induced) subgraphs, the protocol and the analysis are essentially identical to this case\footnote{We rely on the fact that $f$ is a product of multilinear polynomials which holds for not only cliques but also all subgraphs.}. Therefore, we let 
$f(v_1,\ldots,v_k)=\prod_{i\neq j} \widetilde{A}(v_i,v_j)$. Then $f$ is a $N= k\log n$-variate polynomial of total degree $\bigoo(k^2\log n)$.
Taking $|\mathbb F_q|$ to be a prime $q > 2n^k$ (note that the number of $H$'s is at most $\begin{pmatrix}
    n\\k
\end{pmatrix} \leq n^k$), the condition~\eqref{eq:subgraph} is equivalent to
\[
k!\cdot \Delta = \sum_{v_1,\ldots,v_k\in \{0,1\}^{\log n}} f(v_1.\ldots,v_k),
\]
since $f(v_1,\ldots,v_k)$ coincides with $\prod_{i \neq j} A(v_i,v_j)$ 
for all $v_1,\ldots,v_k \in \{0,1\}^{\log n}$. In order to conclude \Cref{theo:naivetriangle} from \Cref{theo:zksc}, 
it remains only to show how the nodes can evaluate $f$ at a single point, and that the view of each node is statistically close to the uniform distribution.

\paragraph*{Oracle access to $f$.}
Let $(r_1,\ldots,r_N)\in \mathbb F_q^N$ be any point. %Let $v_1^*,\ldots,v_k^*$ be $k$ elements of $\mathbb F_q^{N/k}$, where $v_1^* = r_1\ldots r_{N/k}$, $v_2^* = r_{N/k + 1}\cdots r_{2N/k}$ and so on,  respectively. 
Then, $f(r_1,\ldots , r_N)$ is a product of $\bigoo(k^2)$ terms of the multilinear extension $\widetilde{A}$, meaning that it is sufficient if we can compute $\widetilde{A}$ at a random point.
By definition, for any $v_i^*,v_j^* \in \mathbb F_q^{\lceil \log n \rceil}$, we have
\[
\widetilde{A}(v_i^*,v_j^*)=\sum_{v_i\in\{0,1\}^{\log n}}\sum_{v_j\in\{0,1\}^{\log n}} A(v_i,v_j)\cdot \chi_{v_i^*,v_j^*}(v_i,v_j),
\]
and for each fixed $v_i$, the inner sum can be computed locally by node $v_i$. By summing up the tree (with the help of the prover), the root node can compute $\widetilde{A}(v_i^*,v_j^*)$ for given $v_i^*,v_j^* \in \mathbb F_q^{\lceil \log n \rceil}$.

Finally, to conclude it only remains to prove that the polynomial $f$ associated with \subgraphcounting\ is statistically indistinguishable from the uniform distribution when evaluated on a random point. To this end, we show that for all multilinear polynomials, its distribution is statistically close to the uniform distribution when evaluated on a random point. In specific, we prove the following.

\begin{lemma}\label{lemma:tvofsubcount}
    Given a graph $G=(V,E)$, for each node $v\in V$ define as $\widetilde{A}_v$ the random variable representing the portion of $\widetilde{A}(v_i^*,v_j^*)$ that node $v$ receives after summing-up through its subtree. Then, if $U\sim\Unif(\fieldq)$, it holds that

    \begin{equation*}
        \TVc{\widetilde{A}_v}{U}\leq \bigo{\sqrt{\dfrac{\log n}{q}}}
    \end{equation*}
\end{lemma}

\begin{proof}[Proof of \Cref{lemma:tvofsubcount}]

We use the following standard Fourier analysis technique.

\begin{lemma}[Parseval's theorem]\label{lem:Parseval}
For any $f\colon\mathbb{F}_q\to\mathbb C$ define the Fourier transform
\(\widehat f(\lambda):=\sum_{a\in\mathbb{F}_q }f(a)\psi(\lambda a)\), where $\psi(t)=e^{2\pi i t/q}$ for $t\in \mathbb{F}_q$.
Then
\[
      \sum_{a\in\mathbb{F}_q}\abs{f(a)}^{2}
      =\frac1q
        \sum_{\lambda\in\mathbb{F}_q}\!\abs{\widehat f(\lambda)}^{2}.
\]
\end{lemma}

Let $v\in V$ be any node and fix $Q:=\widetilde{A}_v$.
Let $\delta(a): = \underset{\boldsymbol{x}\in \mathbb F_q^N}{\mathbb{P}}[Q(\boldsymbol{x})=a] - \frac1q$ be the bias of $Q$ on $a\in \mathbb{F}_q$. We write $\mathrm{TV}:=\frac12 \sum_{a\in\mathbb F_q}\abs{\delta(a)}$ the total variation distance between $Q$ and the uniform distribution. The Fourier transform of $\delta$ is given by
\begin{align*}
    \widehat{\delta}(\lambda) &:= \sum_{a\in \mathbb{F}_q}\delta(a)\psi(\lambda a)\\
    &= \sum_a \left[\underset{\boldsymbol{x}}{\mathbb{P}}\left[Q(\boldsymbol{x})=a\right]-\frac1q\right]\cdot \psi(\lambda a)\\
    &= \mathbb E_{\boldsymbol{x}\in\mathbb{F}_q^{N}}
        \psi\bigl(\lambda Q(\boldsymbol{x})\bigr) - \frac1q\sum_a \psi(\lambda a)\\
    &= \mathbb E_{\boldsymbol{x}\in\mathbb{F}_q^{N}}
        \psi\bigl(\lambda Q(\boldsymbol{x})\bigr)
\end{align*}
for $\lambda \in \mathbb{F}_q^\times$, and $\widehat{\delta}(0) = 0$.

Next, we let $\vec{\delta}$ be the $q$-dimensional vector where the $a$-th element is $\lvert \delta(a) \rvert$ and $\vec{1}$ be the all-one vector.
By Cauchy-Schwarz we have 
\begin{align}
    \sum_{a\in \mathbb{F}_q}\lvert \delta(a) \rvert &= \vec{\delta}\cdot \vec{1} \nonumber \\
    &\le \lVert \vec{1} \rVert_2 \cdot \lVert \vec{\delta} \rVert_2 \nonumber \\
    &= \sqrt{q}\cdot \sqrt{\sum_a \lvert \delta(a) \rvert^2}. \label{eq:cauchy-schwarz}
\end{align}
Hence, combined with Parseval's theorem (\Cref{lem:Parseval}), the total variation distance is now bounded by
\begin{align*}
    \mathrm{TV} = \frac12 \sum_a \lvert \delta(a) \rvert \le \frac{\sqrt{q}}2  \sqrt{\sum_a \lvert \delta(a) \rvert^2} \\
    = \frac{\sqrt{q}}2  \sqrt{\frac1q \sum_{\lambda \neq 0} \lvert \widehat{\delta}(\lambda) \rvert^2}.
\end{align*}

Our goal is to derive an explicit universal bound on the absolute Fourier coefficients $|\widehat{\delta}(\lambda)|$.

\begin{lemma}
Assume that for some $i$, $Q(x_{1},\ldots, x_N)=x_ih(x_{\neq i})+R(x_{\neq i})$ with $h\not\equiv0$ multilinear.
Then for every $\lambda\in \mathbb{F}_q^\times$
\[
      |\widehat{\delta}(\lambda)|\;\le\;\frac{N-1}q.
\]
\end{lemma}

\begin{proof}
Without loss of generality, assume that $P(x_{1},\ldots, x_N)=x_1h(x_{2},\ldots, x_N)+R(x_2,\ldots,x_N)$ with $h\not\equiv0$ multilinear.
$\widehat{\delta}(\lambda)$ can be written as

\begin{align*}
\widehat{\delta}(\lambda)
 &= \mathbb E_{x_{1},\ldots, x_N}
        \psi\bigl(\lambda P(x_{1},\ldots, x_N)\bigr)\\
 &= \mathbb E_{x_{2},\ldots, x_N}
     \psi\bigl(\lambda R(x_{2},\ldots, x_N)\bigr)
     \Bigl[
        \tfrac1q
        \sum_{x_1\in\mathbb{F}_q}
           \psi\bigl(\lambda x_1 h(x_{2},\ldots, x_N)\bigr)
     \Bigr].
\end{align*}
For fixed $\alpha:=h(x_{2},\ldots, x_N)$ the inner sum equals
\[
   \frac1q
   \sum_{x_1}\psi(\lambda\alpha x_1)
   =\begin{cases}
       1 &(\alpha=0),\\
       0 &(\alpha\ne0),
     \end{cases}
\]
by orthogonality of roots of unity.
Since $h$ involves at most $N-1$ variables and each with degree $1$, Schwartz-Zippel gives

\begin{align*}
    \lvert \widehat{\delta}(\lambda) \rvert
      &\le \underset{x_{2},\ldots, x_N}{\mathbb{P}} \bigl[h(x_{2},\ldots, x_N)=0\bigr]\\ 
      &\le \frac{N-1}q.
\end{align*}
\end{proof}
Recall that for the multilinear extension of the adjacency matrix, there exists at least one variable $x_i$ satisfying the assumption of the above lemma: $Q(x_{1},\ldots, x_N)=x_ih(x_{\neq i})+R(x_{\neq i})$ with $h\not\equiv0$ multilinear (otherwise, the corresponding graph cannot have any edges).
We now get the upper bound for TV:
\[
\mathrm{TV}\leq \frac{\sqrt{q}}2 \cdot\sqrt{\frac1q\cdot(q-1)\frac{(N-1)^2}{q^2}} \le O\left( \frac {\log n} {\sqrt{q}} \right).
\]

\end{proof}

\end{proof}

\section{Reducing Round Complexity by Divide-and-Conquer Sumcheck}

\label{alg:fold-DCS}
\newcommand{\DCS}[2]{\mathsf{Fold}\text{-}\mathsf{DCS}_{#1,#2}}

In this section, we improve the round complexity of subgraph counting using the recent technique of~\cite{levrat2025divide}. For $k\in\bigoo(1)$, this reduces the round complexity from $\bigoo(\log n)$ to $\bigoo(\log \log n)$.  We prove the following theorem.

\begin{theorem}[Restatement of~\Cref{theo:naivetriangle}]\label{theo:divide-and-conquer}
The language  $\subgraphcounting$ is in $$\mathrm{dStatZK}\left[\bigo{k\log k + \log \log n},\bigo{\frac{\log^{2} q}{\log k} + k\log k\log q}, \mathcal C[\bigo{1},k\cdot\mathrm{polylog}(n,k)]\right]$$
    where $q$ is any prime such that $q\in n^{\omega(1)}$ for $k\in\bigoo(1)$ and $q>n^k$ for $k\in\omega(1)$.  
 
\end{theorem}
In order to prove this result, we use the following protocol which we call $\DCS{F}{a}$~\cite{levrat2025divide}.

\begin{tcolorbox}[title =$\DCS{F}{a}$.,enhanced,breakable=true]
\small
\textbf{Assumption:} Both players have oracle access to a polynomial $F:\mathbb F_q^N \to \mathbb F_q$ where $N=2^m$.
\begin{itemize}
\item \textbf{Commit phase:}\\
For $i\in\{1,\ldots,m\}$ do the following:\\
Define $2^{m-i}$-variate polynomial $F_0^{(i)}(x)=\sum_{\vec{a}\in\{0,1\}^{2^{m-i}}}F^{(i-1)}(x,\vec{a})$.
The verifier independently samples $\alpha^{(i)} \in_R \fieldq^{m-i}$ and $z^{(i)}\in_R\fieldq$ uniformly at random, and sends them to the prover.
\\
The prover and the verifier set the following polynomials and value.
\begin{itemize}
    \item $F_1^{(i)}=F^{(i-1)}(\alpha^{(i)},\cdot)$
    \item $F^{(i)} = z^{(i)}F_0^{(i)} + F_1^{(i)}$
    \item $a^{(i)} = z^{(i)}a^{(i-1)}+F_0^{(i)}(\alpha^{(i)})$
\end{itemize} 

\item \textbf{Query phase:}\\
The verifier computes
\[
a^{(m)}=\prod_{j=1}^m z^{(j)} a + \sum_{j=1}^m\left( \prod_{\ell=j+1}^m z^{(\ell)} \right) F_0^{(j)}(\alpha^{(j)})
\]
by querying $ F_0^{(j)}(\alpha^{(j)})$ for $j\in\{1,\ldots,m\}$.\\

The verifier picks $\beta \in \mathbb F$ uniform randomly and checks 
\[
F^{(m)}(\beta) = F(\alpha^{(1)},\ldots,\alpha^{(m)},\beta) + \sum_{j=1}^m z^{(j)}F_0^{(j)}(\alpha^{(j+1)},\ldots,\alpha^{(m)},\beta)
\]
by querying $F^{(m)}(\beta),F(\alpha^{(1)},\ldots,\alpha^{(m)},\beta)$, and $F_0^{(j)}(\alpha^{(j+1)},\ldots,\alpha^{(m)},\beta)$ for $j\in\{1,\ldots,m\}$.
\end{itemize}

\end{tcolorbox}

The protocol is originally implemented in the interactive oracle proof model where each message from the prover is a function which can be queried by the verifier. In the commit phase of $\DCS{F}{a}$, the verifier receives oracle access to $F^{(m)}$ and $F_0^{(j)}$ for $j\in\{1,\ldots, m\}$, and in the query phase the verifier makes $2m + 2$ oracle queries. The following analysis is used in our proof.

\begin{proposition}[Completeness and soundoness of $\DCS{F}{a}$, from \cite{levrat2025divide}]
    If $\sum F =a$, then given an honest prover described in the protocol $\DCS{F}{a}$, the verifier accepts with probability 1,
    Otherwise, the verifier accepts with probability at most
    \[
    \frac{(m+1)(d+1)}{|\fieldq|},
    \]
    where $d$ is the total degree of $F$.
\end{proposition}

\subsection{Splitting the instance} 
\label{alg:split}
Let us consider the $N$-variate polynomial $f$ associated with $\subgraphcounting$, and let $\ell = \lfloor \frac{\log n}{\log k} \rfloor$. We define a protocol $\mathcal{P}_{\mathsf{split}}$ for $\subgraphcounting$ based on splitting the Sumcheck relation, in which the honest prover and the verifier proceed as follows.

\begin{tcolorbox}[title=$\mathcal{P}_{\mathsf{split}}$.,enhanced,breakable=true]
\small
\begin{enumerate}
    \item The prover sends \[
            h_1(x_1,\ldots,x_{\ell}) = \sum_{\mathbf{a}\in\{0,1\}^{N-\ell}} f(x_1,\ldots,x_\ell,\mathbf{a}).
            \]
    
    \item 
    Let $t$ be a positive integer that will be determined later.
    For $i \in \{1,\ldots, t - 1\}$, both do the following.
    \begin{enumerate}
        \item The verifier picks $\boldsymbol{\alpha}^{(i)} \in \mathbb F^{\ell}$. Both set $\widetilde{h}_i(x_1,\ldots,x_{N-i\ell}) = \widetilde{h}_{i-1}(\boldsymbol{\alpha}^{(i)},x_1,\ldots,x_{N-i\ell})$ (where $\widetilde{h}_0 =f$). The verifier computes $a_i=h_i(\boldsymbol{\alpha}^{(i)})$ by querying to $h_i$.
        \item The prover sends \[
            h_{i+1}(x_1,\ldots,x_{\ell}) = \sum_{\mathbf{a}\in\{0,1\}^{N-(i+1)\ell}} \widetilde{h}_i(x_1,\ldots,x_\ell,\mathbf{a}).
            \]
    \end{enumerate}
    \item The verifier picks $\boldsymbol{\alpha}^{(t)} \in \mathbb F^{\ell}$. The prover sends $\widetilde{h}_t(x_1,\ldots,x_{N-t\ell}) = \widetilde{h}_{t-1}(\boldsymbol{\alpha}^{(t)},x_1,\ldots,x_{N-t\ell})$. Both set $a_t=h_t(\boldsymbol{\alpha}^{(t)})$. The verifier picks $\boldsymbol{\beta} \in \mathbb F^{N-t\ell}$, and checks 
    \[
    \tilde{h}_t(\boldsymbol{\beta}) = f(\boldsymbol{\alpha}^{(1)},\ldots,\boldsymbol{\alpha}^{(t)},\boldsymbol{\beta})
    \]
    by querying to $\tilde{h}_t$ and $f$.
    \item The prover and the verifier jointly solve $\DCS{h_1}{a}$, $\DCS{h_{i+1}}{a_{i}}$ for $i \in \{1,\ldots, t-1\}$, $\DCS{\widetilde{h}_t}{a_{t}}$ in parallel.
\end{enumerate}
\end{tcolorbox}

For this protocol, we prove the following.
\begin{proposition}[Completeness]\label{proposition:completeness}
    If $\sum f=a$, then given $h_1,\ldots,h_t,\widetilde{h}_t$ by the honest prover of $\mathcal{P}_{\mathsf{split}}$, the verifier accepts with probability 1.
\end{proposition}
\begin{proof}
    Assume that $\sum_{x_1,\ldots,x_{N}} f(x_1,\ldots,x_N) =a$.
    The verifier accepts $\DCS{h_1}{a}$ since
    \[
        \sum_{x_1,\ldots,x_{\ell}} h_1(x_1,\ldots,x_{\ell}) = \sum_{x_1,\ldots,x_{N}} f(x_1,\ldots,x_N   ) = a.
    \]
    For $i\in \{1,\ldots ,t-1\}$, the verifier accepts $\DCS{h_{i+1}}{S_i}$ since
    \begin{align*}
        \sum_{x_1,\ldots,x_{\ell}} h_{i+1}(x_1,\ldots,x_{\ell}) &= \sum_{x_1,\ldots,x_{\ell}} \sum_{\mathbf{a}\in\{0,1\}^{N-(i+1)\ell}} \widetilde{h}_i(x_1,\ldots,x_\ell,\mathbf{a})\\
        &=   \sum_{x_1,\ldots,x_{N-i\ell}} \widetilde{h}_{i-1}(\boldsymbol{\alpha}^{(i)}, x_1,\ldots,x_{N-i\ell})\\
        &= h_i(\boldsymbol{\alpha}^{(i)}).
    \end{align*}
    The verifier accepts $\DCS{\widetilde{h}_t}{a_t}$ since
    \begin{align*}
        \sum_{x_1,\ldots,x_{N-t\ell}} \widetilde{h}_t(x_1,\ldots,x_{N-t\ell}) 
        &= \sum_{x_1,\ldots,x_{N-t\ell}} \widetilde{h}_{t-1}(\boldsymbol{\alpha}^{(t)},x_1,\ldots,x_{N-(t-1)\ell}) \\
        &= h_{t}(\boldsymbol{\alpha}^{(t)}).
    \end{align*}
    Finally, it holds
    \[
    \tilde{h}_t(\boldsymbol{\beta}) = f(\boldsymbol{\alpha}^{(1)},\ldots,\boldsymbol{\alpha}^{(t)},\boldsymbol{\beta})
    \]
    for any $\boldsymbol{\beta}\in \mathbb F^{N-t\ell}$ as $\tilde{h}_t(\mathbf x) \equiv F(\boldsymbol{\alpha}^{(1)},\ldots,\boldsymbol{\alpha}^{(t)},\mathbf x).$
    
\end{proof}

\begin{proposition}[Soundness]\label{proposition:soundness}
    If $\sum f\neq a$, then for any prover of $\mathcal{P}_{\mathsf{split}}$, the verifier accepts with probability at most $O\left( \frac{td\log N}{|\mathbb F_q|} \right)$, where $d$ is the total degree of $f$.
\end{proposition}
\begin{proof}
    
    Assume that $\sum f\neq a$. Let $p_i$ for $i \in \{1.\ldots, t\}$, $\widetilde{p}_t$ be the polynomial actually sent by the prover. Then, $a_i= p_i(\boldsymbol{\alpha}^{(i)})$ for $i\in\{1,\ldots, t\}$. We write $\mathsf{Fold}\text{-}\mathsf{DCS}_i =\DCS{p_i}{a_{i-1}}$ for $i \in \{1,\ldots, t\}$ and $\mathsf{Fold}\text{-}\mathsf{DCS}_{t+1} = \DCS{\widetilde{p}_t}{a_t}$.

    Define the following events: $\mathcal E_{\mathsf{acc}}$ is the event that the verifier accepts. $\mathcal E_i$ is the event $p_i=h_i$ for $i\in \{1,\ldots,t\}$, and $\mathcal E_{t+1}$ is the event $\widetilde{p}_t=\widetilde{h}_t$. $\mathcal F_i$ is the event $\sum p_i=a_{i-1}$ for $i\in \{1,\ldots ,t\}$, and $\mathcal F_{t+1}$ is the event $\sum \widetilde{p}_t=a_t$. Define $\mathcal F_{\mathsf{all}} = \underset{i}{\bigwedge} \mathcal F_i$. Due to the soundness of $\mathsf{Fold}\text{-}\mathsf{DCS}$, for $i\in\{1,\ldots ,t+1\}$, 
    \[
    \mathrm{Pr}\left[ \text{The verifier accepts $\mathsf{Fold}\text{-}\mathsf{DCS}_i$}\middle| \overline{\mathcal F_i} \right]\le O\left( \frac{d\log N}{|\mathbb F_q|} \right).
    \]
    By union bound,
    \begin{align*}
    \mathrm{Pr}\left[ \mathcal E_{\mathsf{acc}} \middle| \overline{\mathcal F_{\mathsf{all}}} \right]
    &\le \sum_{i=1}^{t+1} \mathrm{Pr}\left[ \mathcal E_{\mathsf{acc}} \middle| \overline{\mathcal F_{i} } \right] \\
    &\le \sum_{i=1}^{t+1} \mathrm{Pr}\left[ \text{The verifier accepts $\mathsf{Fold}\text{-}\mathsf{DCS}_i$}\middle| \overline{\mathcal F_{i}} \right] 
    \le O\left( \frac{td\log N}{|\mathbb F_q|} \right).
    \end{align*}

    Next, we assume $\mathcal F_{\mathsf{all}}$ holds. Then, from the assumption that $\sum f\neq a$, we have
    \[
    \mathrm{Pr}\left[ \overline{\mathcal E_1} \middle| \mathcal F_{\mathsf{all}} \right] = 1.
    \]
    Under $ \overline{\mathcal E_1}\wedge \mathcal F_2$, $p_1 \not\equiv h_1$ and hence $p_1(\boldsymbol{\alpha}^{(1)}) = h_1(\boldsymbol{\alpha}^{(1)})$ with probability at most $d/|\mathbb F_q|$, implying $\overline{\mathcal E_2}$ holds with probability at least $1-d/|\mathbb F_q|$. Therefore,
    \[
    \mathrm{Pr}\left[ \overline{\mathcal E_2} \middle| \mathcal F_{\mathsf{all}} \right] = \mathrm{Pr}\left[ \overline{\mathcal E_2} \middle| \mathcal F_{\mathsf{all}} \wedge \overline{\mathcal E_1} \right] \cdot \mathrm{Pr}\left[ \overline{\mathcal E_1} \middle| \mathcal F_{\mathsf{all}} \right] \ge 1-d/|\mathbb F_q|.
    \]
    Similarly, for $i\in \{2,\ldots, t\}$, under $ \overline{\mathcal E_i}\wedge \mathcal F_{i+1}$, we have $p_i(\boldsymbol{\alpha}^{(i)}) = h_i(\boldsymbol{\alpha}^{(i)})$ with probability at most $d/|\mathbb F_q|$, implying $\overline{\mathcal E_{i+1}}$ holds with probability at least $1-d/|\mathbb F_q|$. Therefore,
    \[
    \mathrm{Pr}\left[ \overline{\mathcal E_{i+1}} \middle| \mathcal F_{\mathsf{all}}  \right] \ge \mathrm{Pr}\left[ \overline{\mathcal E_{i+1}} \middle| \mathcal F_{\mathsf{all}} \wedge \overline{\mathcal E_i} \right] \cdot \mathrm{Pr}\left[ \overline{\mathcal E_i} \middle| \mathcal F_{\mathsf{all}} \right]  \ge 1-\frac{id}{|\mathbb F_q|}.
    \]
    
    Now, under $\overline{\mathcal E_{t+1}}$, 
    \[
    \tilde{p}_t(\boldsymbol{\beta}) = f(\boldsymbol{\alpha}^{(1)},\ldots,\boldsymbol{\alpha}^{(t)},\boldsymbol{\beta})
    \]
    holds with probability at most $D/|\mathbb F|$. Therefore, we have
    \begin{align*}
    \mathrm{Pr}\left[ \mathcal E_{\mathsf{acc}} \middle| \mathcal F_{\mathsf{all}} \right] 
    &\le  \mathrm{Pr}\left[ \tilde{p}_t(\boldsymbol{\beta}) = f(\boldsymbol{\alpha}^{(1)},\ldots,\boldsymbol{\alpha}^{(t)},\boldsymbol{\beta}) \middle| \mathcal F_{\mathsf{all}} \right] \\
    & \le \mathrm{Pr}\left[ \tilde{p}_t(\boldsymbol{\beta}) = f(\boldsymbol{\alpha}^{(1)},\ldots,\boldsymbol{\alpha}^{(t)},\boldsymbol{\beta}) \middle| \mathcal F_{\mathsf{all}} \wedge \overline{\mathcal E_{t+1}} \right] \cdot \mathrm{Pr}\left[\overline{\mathcal E_{t+1}} \middle| \mathcal F_{\mathsf{all}} \right] +  \mathrm{Pr}\left[\mathcal E_{t+1} \middle| \mathcal F_{\mathsf{all}} \right]\\
    &\le \frac{d}{|\mathbb F_q|}\cdot \left(1-\frac{td}{|\mathbb F_q|}\right) + \frac{td}{|\mathbb F_q|} \\
    &\le \frac{(t+1)d}{|\mathbb F_q|}.
    \end{align*}
    The acceptance probability is bounded by
    \[
    \mathrm{Pr}\left[\mathcal E_{\mathsf{acc}}\right] = \mathrm{Pr}\left[ \mathcal E_{\mathsf{acc}} \middle| \mathcal F_{\mathsf{all}} \right] + \mathrm{Pr}\left[ \mathcal E_{\mathsf{acc}} \middle| \overline{\mathcal F_{\mathsf{all}}} \right] = O\left( \frac{td\log N}{|\mathbb F_q|} \right).
    \]
\end{proof}

\subsection{Proof of~\Cref{theo:divide-and-conquer}}
\Cref{proposition:completeness,proposition:soundness} show that $\mathcal{P}_{\mathsf{split}}$ correctly solves \subgraphcounting.
We next explain how to implement oracle access to the polynomials appeared in $\mathcal{P}_{\mathsf{split}}$.
\paragraph*{Implementing oracle access to the polynomials.} 
For $\subgraphcounting$, we use $N=k\lceil \log n \rceil$ since 
\[
f(v_1,\ldots,v_k)=\prod_{i\neq j}\widetilde{A}(v_i,v_j).
\]
We also set $t=k \lceil \log k \rceil$ in $\mathcal{P}_{\mathsf{split}}$.
Observe that $\widetilde{A}(i,j)$ is a multilinear polynomial. So $f$ has individual degree $k-1$.
Each of $h_1,\ldots, h_t,\widetilde{h}_t$ can be obtained by taking sums of $f$ or substituting random values to $f$.
Thus, they also have individual degree $k-1$, implying that
 the number of monomials in each of them is at most $k^\ell = k^{\lfloor\frac {\log n}{\log k}\rfloor} \le n$ (here we used the fact that the number of variables in $\widetilde{h}_t$ is $N-t\ell \le \ell$). The prover can distribute the coefficient of a single monomial to each node, with a vector from $\{0,1,\ldots,k-1\}^\ell$ representing the corresponding monomial (the $i$-th element of the vector represents the degree of the $i$-th variable). To implement queries, the prover gives each node the partial sums of monomials assigned to all its descendants in the spanning tree, and the root node can compute the value by aggregating them along the tree.  
 In all of instances of $\mathsf{Fold}\text{-}\mathsf{DCS}$, polynomials sent by the prover have at most $n$ monomials as well, and therefore oracle access to these polynomials can be implemented in exactly the same way.
\paragraph*{Analysis of complexities.}
The first three steps of $\mathcal{P}_{\mathsf{split}}$ require $t = \bigoo(k\log k)$ rounds and $\ell q \in\bigoo(\frac{\log^2 q }{\log k})$-bit messages. The last step (Step (4)) requires $\bigoo(\log \ell) = \bigoo(\log \log n)$ rounds and $\bigoo(t\log q) = \bigoo(k\log q \log k)$-bit messages. 

 The simulator for this protocol simulates each oracle query by an uniform random value. Recall that all polynomials defined in the protocol $\mathcal{P}_{\mathsf{split}}$ are also multilinear. Thus, the Fourier analysis given in~\Cref{lemma:tvofsubcount} shows that the evaluation of these polynomials at random points are statistically close to the uniform distribution. For the first three steps of $\mathcal{P}_{\mathsf{split}}$, the simulator can be taken from the class $\mathcal{C}[\bigoo(1),\bigoo(k\log^2 q)]$. For the last step, recalling that the protocol $\DCS{F}{a}$ for $N$ variables requires $\bigoo(\log N)$ rounds and $\bigoo(\log N \log q)$-bit messages, it can be simulated by the class $\mathcal{C}[\bigoo(1),\bigoo(t\log^2 \ell \log q)]$. The resulting simulator can be taken from the class $\mathcal{C}[\bigoo(1),k\cdot \mathrm{polylog}(n,k)]$, which completes the proof of~\Cref{theo:divide-and-conquer}.

\section{A Barrier for Round Reduction for \nonkcol}
In this section, we again deal with \nonkcol, focusing on constant-degree graphs. In \Cref{subsec:upper-bound-constant-degree-graphs}, we prove the upper bound side of~\Cref{theorem:conditional-lower-bound}, and in \Cref{subsec:lower-bound-constant-degree-graphs}, we prove the lower bound side of~\Cref{theorem:conditional-lower-bound}.

\subsection{Upper Bound for Constant-Degree Graphs}\label{subsec:upper-bound-constant-degree-graphs}
Consider the task of verifying \nonkcol\ of constant-degree graphs.
We use the polynomial $f$ for \nonkcol, i.e., $N= \bigoo(n)$ (for constant integer $k$). We run the first three steps of $\mathcal{P}_{\mathsf{split}}$ in \Cref{alg:split}, using $t$ the minimum integer such that $N-t\ell \le \ell$ where $\ell=\lceil \log_3 n \rceil$, to generate $t+1$ instances of Sumcheck. Instead of running $\mathsf{Fold}\text{-}\mathsf{DCS}$ as in the final step of $\mathcal{P}_{\mathsf{split}}$, the prover and the verifier perform the standard Sumcheck (using our Sumcheck compiler) for the generated instances. This solves $\nonthreecol$ of constant-degree graphs.

\paragraph*{Analysis.} Generating $t+1$ instances by the first three steps of $\mathcal{P}_{\mathsf{split}}$ takes $\bigoo(t)\in\bigoo(n/\log n)$ rounds and $\bigoo(\log^{2+o(1)} n)$-bit messages per node per round. Running $t+1$ instances of Sumcheck requires the same amount of rounds and message size: Consider that the prover and the verifier check $\bigoo(\log n)$ Sumcheck instances at once, using $\bigoo(\log n)$ rounds and $\bigoo(\log^{2+o(1)} n)$-bit messages per node per round. This is repeated for $\bigoo(t/\log n) = \bigoo(n/\log^2 n)$ times, making the total round complexity $\bigoo(n/\log n)$.
The completeness and soundness follow from \Cref{proposition:completeness} and \Cref{proposition:soundness}. But here the total degree in the soundness error is replaced by the individual degree, which is $\bigoo(1)$, as we use the standard Sumcheck in this case.

\subsection{Conditional Lower Bound for Constant-Degree Graphs}\label{subsec:lower-bound-constant-degree-graphs}

Now we show a lower bound for $\nonthreecol$ problem for simplicity. The same reduction holds for any \nonkcol\ problem for fixed $k$. We will use the following lemmas.
\begin{lemma}[The Sparsification Lemma~\cite{impagliazzo2001problems}]\label{lem:sparsification}
    For all $\varepsilon\geq 0$, $k$-CNF $F$ on $n$ variables can be written as the disjunction of at most
$2^{\varepsilon n}$ $k$-CNFs $\{F_i\}$ on $n$ variables such that $F_i$ contains each variable in at most $\mathrm{poly}(1/\varepsilon)$ clauses. Moreover, this reduction takes at most
$poly(n) 2^{\varepsilon n}$ time.
\end{lemma}
Let $\mathsf{AM}[b,M]$ be the class of languages that can be decided by $\mathsf{AM}$ protocols in which the prover and the verifier exchange at most $b$ bits in total, and the number of messages exchanged between the prover and the verifier is $M$. Let $\mathsf{AMTIME}[T]$ be the class of languages that can be decided by 2-round $\mathsf{AM}$ protocols in which the prover and the verifier send each at most $T$ bits, and the runtime of the verifier is at most $T$. The following round reduction lemma for Arthur-Merlin protocols is shown in~\cite{goldreich2002interactive}.
\begin{lemma}[\cite{goldreich2002interactive}]\label{lem:parallelization}
    $\mathsf{AM}[b,M] \subseteq \mathsf{AMTIME}[(b\cdot M)^{\bigoo(M)}]$.
\end{lemma}

We are now ready to prove the following theorem.
\begin{theorem}
    Let $\Pi$ be the problem of certifying $\nonthreecol$ of constant degree graphs.
    Assuming $\mathsf{coNP}\not\subseteq \mathsf{AMTIME}[2^{o(n)}]$, for any $b,M>0$ satisfying $b=\mathrm{poly}(n)$ and $M~\in~o(n/\log n)$, there is no $M$-round and $b$-message distributed Arthur-Merlin protocol for $\Pi$ where the verification algorithm at each node is a polynomial-time algorithm.
\end{theorem}

\begin{proof}
Let $\Phi$ be a $3$-CNF on variables $x_1,\ldots,x_n$. Fix arbitrary $\varepsilon > 0$.
By Lemma~\ref{lem:sparsification}, in $2^{\bigoo(\varepsilon n)}$ time, we can construct $2^{\varepsilon n}$ $3$-CNFs $\{\phi_j\}_{j\in [2^{\varepsilon n}]}$ such that each CNF $\phi_j$ contains $m = \bigoo(n)$ clauses such that each variable $x_i$ appears at most $\mathrm{poly}(1/\varepsilon) = \bigoo(1)$ clauses. Fix arbitrary CNF $\phi$ in these $2^{\varepsilon n}$ $3$-CNF formulas.

\paragraph*{The standard reduction from 3-SAT to 3-coloring.}
Let $\phi$ be a 3-CNF formula on $n$ variables, where each variable appears at most $\Delta$ clauses for some constant $\Delta$. We follow the standard reduction from 3-SAT to 3-coloring (e.g., \cite{goldreich2008computational}).
The constructed graph $G_\phi$ has three special nodes $v_T,v_F,v_B$ which form a triangle.
For each variable $x_i$, we have two nodes $v(x_i),v(\overline{x_i})$. We have an edge connects them, and they are also connected to $v_B$. For each clause $C$, we create a clause gadget, as illustrated in Figure~\ref{fig:gadget}. A clause gadget for a clause $C=(x_i\lor x_j \lor x_k)$ contains six nodes where $v(x_i),v(x_j),v(x_k), v_F, v_B$ are connected to these gadget nodes.  
\begin{figure}[h]
    \centering
    \includegraphics[width=0.7\linewidth]{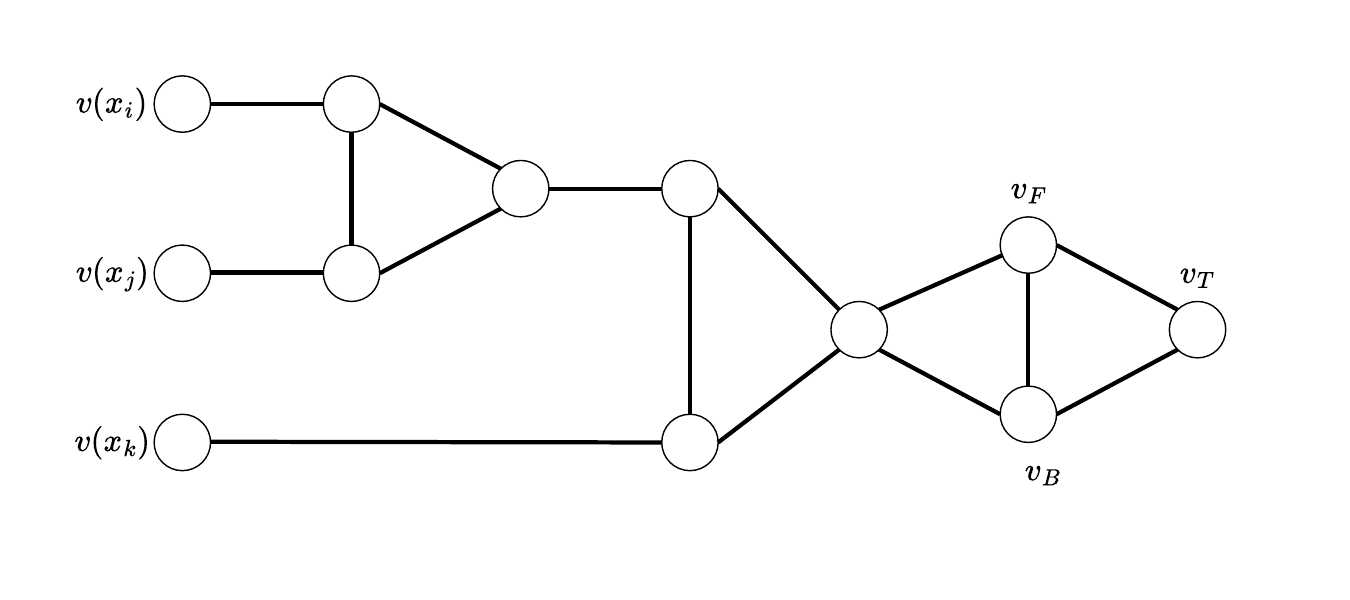}
    \caption{Clause gadget for a clause $C=(x_i\lor x_j \lor x_k)$. In any valid 3-coloring, at least one of $v(x_i),v(x_j),v(x_k)$ has the same color as of $v_T$. Note that all of $v(x_i),v(x_j),v(x_k)$ are connected to $v_B$ (but the corresponding edges are omitted in the figure for simplicity).}
    \label{fig:gadget}
\end{figure}

\paragraph*{Reducing the maximum degree of the graph.}
In the above construction, the number of nodes is $\bigoo(n+m)$ where $m$ is the number of clauses in $\phi$. Since we assumed that each variable appears at most constant number of clauses, $m\in\bigoo(n)$. In the constructed graph, variable nodes and gadget nodes have constant degrees, but the degree of $v_T,v_F,v_B$ are at least $\Omega(n)$. We can reduce the maximum degree as follows. Let us focus on one node, e.g., $v_B$. We add a binary tree rooted at $v_B$ of depth $\bigoo(\log(\mathrm{deg}(v_B)))$, where the number of leaves is exactly equal to $\mathrm{deg}(v_B)$. We delete each incident edge of $v_B$, and add a new edge connected to an unused leaf node instead of the original $v_B$. Finally, we replace each tree edge by another gadget as illustrated in Figure~\ref{fig:tree-gadget}, to make all tree nodes have the same color. 
\begin{figure}[h]
    \centering
    \includegraphics[width=0.4\linewidth]{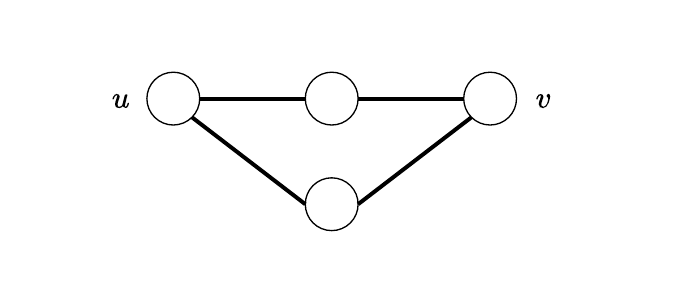}
    \caption{A gadget for a tree edge $(u,v)$. In any valid 3-coloring, $u$ and $v$ must have the same color.}
    \label{fig:tree-gadget}
\end{figure}
This modification increases the total number of nodes to $\bigoo(n \log n)$ while ensuring maximum degree $\bigoo(1)$.

For a given $3$-SAT formula $\phi$, we construct $G_\phi$ in $\mathrm{poly}(n)$ time. Assume that \nonthreecol\ can be certified by an $m$-round $\mathsf{dAM}$ protocol with message size $b$ on constant-degree, $n$-node and $m$-edge graphs. Recall that $G_\phi$ is a $\bigoo(n\log n)$-node $\bigoo(n\log n)$-edge graph. The verification algorithm $\mathcal{A}$ of this $\mathsf{dAM}$  protocol is a poly-time algorithm. Thus, simulating the verification algorithm $\mathcal{A}$ of all nodes gives an $\mathsf{AM}[b',m]$ protocol for some $b'=\mathrm{poly}(n)$. Using Lemma~\ref{lem:parallelization}, this can be converted to an $\mathsf{AMTIME}[2^{o(n)}]$ protocol, as
\[
 (b'\cdot m)^{\bigoo(m)} = \mathrm{poly}(n)^{o\left(\frac{n}{\log n}\right)} = 2^{o(n)}.
\]
Running the converted protocol for $2^{\varepsilon n}$ different instances (created by Lemma~\ref{lem:sparsification}) solves \textsc{Unsat} from the construction of $G_{\phi}$. This implies that for any $\varepsilon > 0$, we have 
$
\textsc{Unsat} \in \mathsf{AMTIME}[2^{\varepsilon n}\cdot 2^{o(n)} \cdot n^{\bigoo(1)}]
$ for any $\varepsilon >0$, implying
$
\textsc{Unsat} \in \mathsf{AMTIME}[2^{o(n)}].
$
\end{proof}

\section{Conclusion}

In this work, based on the robust Sumcheck protocol, we introduce a general framework for studying classical \emph{hard} problems in distributed verification under a statistical zero-knowledge guarantee. Using this framework, we improve the state of the art for two central problems---\nonkcol\ and \subgraphcounting.

Our results leave open the problem of understanding the scope of our Sumcheck-based approach.  In our application to $\nonkcol$, the reduction from $k$-coloring to $k$-SAT preserves the relevant locality: the
resulting constraints are associated with vertices and edges, and the corresponding arithmetization can be evaluated by local computation plus aggregation.  It would be interesting to develop a more systematic theory
of such locality-preserving reductions between distributed graph problems.\footnote{A related notion of reductions has recently been studied in local certification (non-interactive distributed verification model) by Esperet and Zeitoun~\cite{esperet2025reductions}, where local reductions are used to transfer proof-size lower bounds between graph properties. Our question is different, but similar in spirit: we ask which reductions preserve the locality and algebraic structure needed for distributed Sumcheck.}
More generally, one may ask which distributed graph predicates admit a low-degree arithmetization whose terms are locally computable and whose global sum captures the desired property. Such a characterization would
clarify a broader class of Sumcheck-friendly distributed problems. Relatedly, our results suggest that round compression is strongly problem-dependent; it remains open to characterize when the round complexity of distributed Sumcheck can be reduced, as in our protocol for $\subgraphcounting$.

A second direction is to make the framework more black-box. Our general distributed Sumcheck protocol assumes oracle access to the target polynomial, but in concrete graph problems this oracle access must be implemented inside the network, and the resulting transcript must be shown to preserve the statistical zero-knowledge property.  In this work, this step is handled separately for $\nonkcol$ and $\subgraphcounting$, using
problem-specific properties of their arithmetizations.  A natural goal is to abstract this part into a general compiler: given a suitable arithmetization of a distributed predicate, can one automatically obtain a
distributed statistical zero-knowledge proof with comparable round and message complexity? Such a compiler would be interesting in light of existing compilers for distributed zero-knowledge~\cite{zkdef, grilo2025distributed}.  These results provide broad applicability, but they may incur large proof-size overheads or rely
on cryptographic assumptions. A Sumcheck-based compiler could offer a complementary route for problems admitting sufficiently local low-degree arithmetizations. 
A stronger version of this question is whether such a compiler can achieve {\it perfect} zero-knowledge for some nontrivial classes of problems, namely, whether the oracle evaluation part can be simulated exactly rather than only up to small statistical distance. 

{\bf Acknowledgments.} The authors are grateful to the anonymous reviewers for their constructive comments and suggestions, which improved the current presentation of this work.

\bibliographystyle{alpha}
\bibliography{bibliography}

\end{document}